\pdfoutput=1

\documentclass[journal,11pt,onecolumn]{IEEEtran}

\usepackage{hyperref}
\usepackage{graphicx}
\usepackage[cmex10]{amsmath}
\usepackage{amsfonts}
\usepackage{cite}

\usepackage[linesnumbered,ruled]{algorithm2e}
\SetAlFnt{\small}
\SetAlCapFnt{\small}
\SetAlCapNameFnt{\small}
\usepackage{algorithmic}
\algsetup{linenosize=\tiny}

\usepackage{dsfont, color, soul}
\usepackage{amsmath}
\usepackage{amssymb}
\usepackage{booktabs}
\usepackage{amsthm}

\usepackage{siunitx}

\usepackage{textcomp}

\usepackage[none]{hyphenat} % disable hyphenation

\DeclareMathOperator*{\argmax}{arg\,max}
\DeclareMathAlphabet\mathbfcal{OMS}{cmsy}{b}{n}

\newtheorem{theorem}{Theorem}

\newtheorem{definition}{Definition}

\newcommand*\tageq{\refstepcounter{equation}\tag{\theequation}}

\allowdisplaybreaks % allow breaks in "align" environments

\newcommand{\MAP}{{\scriptscriptstyle\textup{MAP}}}

\ifCLASSINFOpdf
\else
\fi

\begin{document}

\setlength{\textfloatsep}{12pt}
\setlength{\abovecaptionskip}{-1pt}
\setlength{\belowcaptionskip}{-1pt}

\setlength{\parskip}{1pt}
\setlength{\parsep}{1pt}
\setlength{\headsep}{1pt}
\setlength{\topskip}{1pt}
\setlength{\topsep}{1pt}
\setlength{\partopsep}{1pt}

\title{Geo-spatial Location Spoofing Detection for Internet of Things}

\author{\IEEEauthorblockN{Jing Yang Koh, Ido Nevat, Derek Leong, and Wai-Choong Wong \thanks{A shorten version of this work
has been accepted to the IEEE IoT Journal (IoT-J) on 08-Feb-2016. 
}
\thanks{J.~Y.~Koh, I.~Nevat, and D.~Leong are with the Institute for Infocomm Research (I\textsuperscript{2}R), Singapore. 
W.-C.~Wong is with the Department of Electrical and Computer Engineering, National University of Singapore.}
\thanks{The work of J.~Y.~Koh was supported in part by the Agency for Science,
Technology and Research (A*STAR) Graduate Scholarship.
The work of I.~Nevat was supported in part by A*STAR, Singapore, under SERC Grant 1224104048.
The work of D.~Leong was supported in part by A*STAR, Singapore, under SERC Grant 1224104049.} 
}

}

\maketitle

\begin{abstract}
We develop a new location spoofing detection algorithm for geo-spatial tagging and location-based services in the Internet of Things (IoT), called Enhanced Location Spoofing Detection using Audibility (ELSA) which can be implemented at the backend server without modifying existing legacy IoT systems.
ELSA is based on a statistical decision theory framework and uses two-way time-of-arrival (TW-TOA) information between the user's device and the anchors.
In addition to the TW-TOA information, ELSA exploits the implicit available audibility information to improve detection rates of location spoofing attacks.
Given TW-TOA and audibility information, we derive the decision rule for the verification of the device's location, based on the generalized likelihood ratio test.
We develop a practical threat model for delay measurements spoofing scenarios, and investigate in detail the performance of ELSA in terms of detection and false alarm rates.
Our extensive simulation results on both synthetic and real-world datasets demonstrate the superior performance of ELSA compared to conventional non-audibility-aware approaches.
\end{abstract}


\IEEEpeerreviewmaketitle

\section{Introduction}
Wireless localization has been an active research topic in the last decade due to its significance in many existing applications.
In particular, the area of detecting  \emph{location spoofing} attempts has become increasingly important.
This is due to its key role in proliferating applications such as location-based services \cite{hasan2015, yan2012}, intelligent transport systems \cite{yan2014, yan20142, yan2015,  mala2014}, mobile and ad hoc networks \cite{fiore2013, yan2015}, wireless sensor networks \cite{neal2005, capkun2008, wei2013}, and other mission-critical systems \cite{misra2009}.
With the expansion of the Internet of Things (IoT)~\cite{yang2015}, more and more users are expecting reliable and trustworthy estimates of the locations of the ``things'' in their systems.
Without reliable information, location-based services may be severely disrupted, causing inconvenience to end users or even resulting in the loss of human lives especially in hazardous applications.
In fact, high accuracy and precision are key requirements in many IoT applications today~\cite{yang2015}.

Spatially deployed \emph{anchors} (or reference nodes) can be used to estimate the distance of targets in the range-based time-of-arrival (TOA) localization techniques~\cite{neal2005, guvenc2009, 802153}.
Specifically, we focus on the TOA-based two-way ranging (TWR) protocol~\cite{neal2005,  guvenc2009, 802153} where a \emph{target} (user device or tag) simply needs to reply to range request packets sent from the anchors.
This enables the anchors to estimate their distances from the target by making use of the time of flight (delay) information.
However, a malicious target can attempt to spoof its location by affecting the delay measurements received by the anchors.
Therefore, many location spoofing detection schemes~\cite{yan2012, yan2014, yan2015,  mala2014, fiore2013, capkun2008, wei2013, chiang2012, basilico2014} have been proposed to deal with this threat.
Typically, the detection system uses trilateration (or multilateration) \cite{neal2005, misra2009} to fuse three or more distance estimates to localize a node in two dimensions~\cite{fiore2013, capkun2008, guvenc2009, basilico2014}, reducing ambiguity in the location estimates.

However, we show that this fundamental requirement of distance estimates from at least three audible anchors can be relaxed --- localization can often be done reliably with fewer \emph{audible} anchors.
Two nodes A and B are said to be audible to each other if they are able to successfully decode the transmitted signals from each other.
In the context of this paper, the received signal strength has to be above a predefined threshold in order for them to be audible.
This will be formally defined in Definition 2.
In contrast to prior works that simply ignore inaudible anchors (e.g., the inaudible anchors are excluded from the trilateration calculations), we exploit the implicit \emph{inaudibility} (or outage)  information to improve the location spoofing detection rate at essentially no additional cost.

Using the concept of audibility, we develop a generalized likelihood ratio test (GLRT)~\cite{neyman1933} called Enhanced Location Spoofing Detection using Audibility (ELSA) to detect location spoofing attacks.
The statistical GLRT hypothesis testing technique is a well-recognized approach that can be applied to the received TOA delay measurements to distinguish an honest target from a malicious target.
We choose TOA-based localization as it is widely used (e.g., in Global Positioning System (GPS)) and provides the best accuracy (e.g., in the range of centimeters for ultra-wide band (UWB) devices~\cite{ dardari2009, wym2012}) compared to other range-based (e.g., received signal strength (RSS)) and range-free approaches~\cite{he2003}.
We also consider GPS-denied indoor or urban environments where the GPS measurements are not readily available~\cite{guvenc2009}.
We then study the effectiveness of ELSA under adversarial settings and show that it significantly outperforms the conventional non-audibility-aware TOA-based approaches (e.g., \cite{yan20142}, which adopts a similar likelihood ratio test approach but does not consider audibility in its likelihood probability functions).

\begin{figure}
  \centering
  \includegraphics[width=3in]{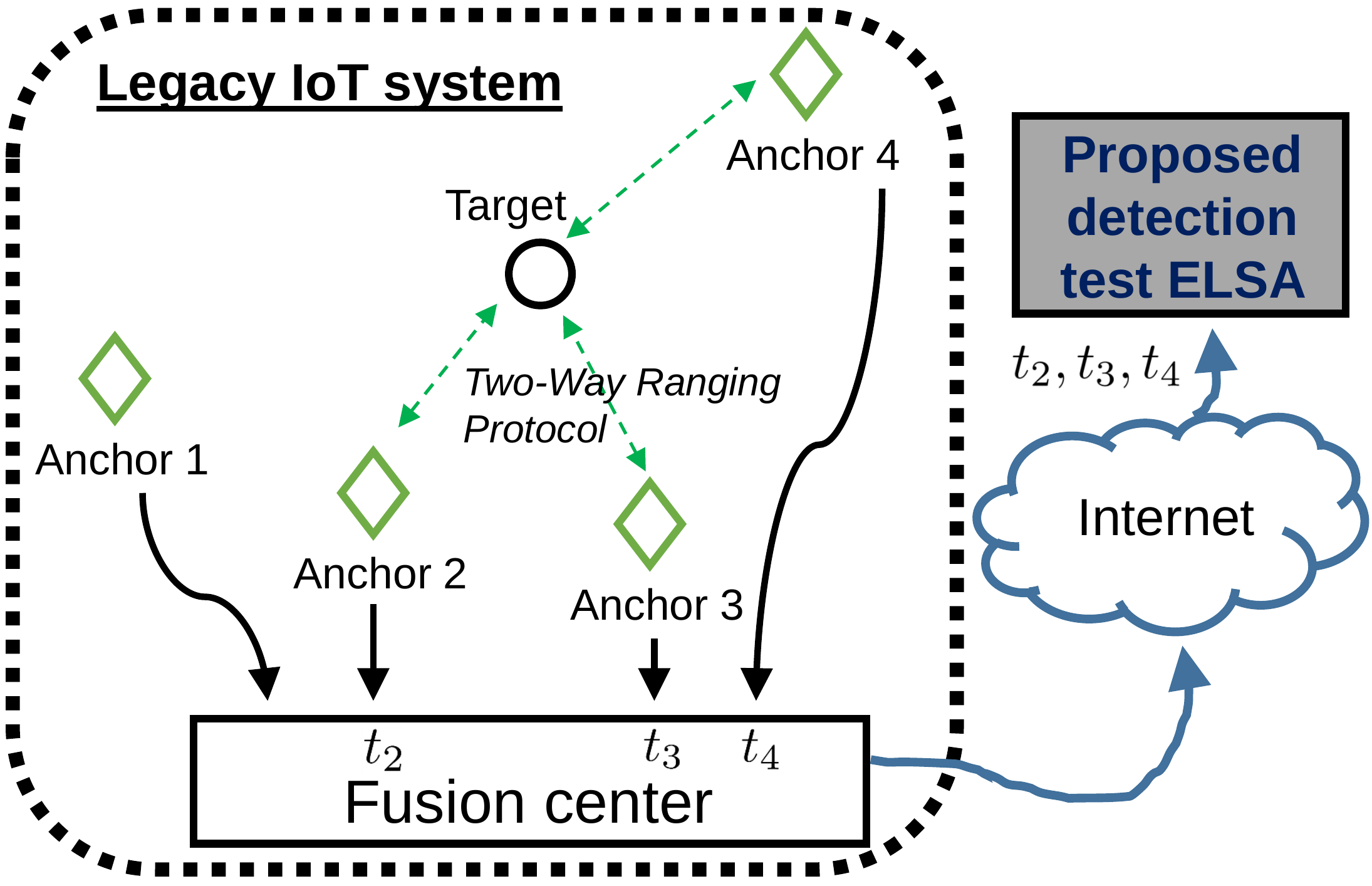}
  \caption{System model with a target, multiple anchors, and a fusion center. The proposed detection test can be implemented at the backend server which receives the TOA delay measurements from the anchors via the fusion center, without changing the legacy IoT system.}
  \label{fig:system_model}
\end{figure}

ELSA can be applied to a wide range of legacy IoT systems and to emerging applications to improve localization and the detection of spoofing attempts at essentially no additional cost because it uses implicit audibility information available in conventional TOA-based localization systems.
This allows our approach to be implemented solely at the backend server without changing existing client IoT devices or network communications protocols (see Fig.~{\ref{fig:system_model}}).
ELSA is particularly beneficial when low-level information (e.g., RSS readings from the radio module) is not available because of device limitations, or when devices have limited resources for running computationally expensive cryptographic operations.
An example of such a use case is in the tagging of physical objects (high-value assets, equipment, luggage, personal wallets, etc.) to facilitate easy retrieval (e.g., see\cite{pixieweb, decawaveweb}).
Without a location spoofing detection test like ELSA, an adversary may successfully steal a tagged item without detection.
Another potential application is in city tagging~\cite{iotweb}, where users can virtually tag places or objects and add a description of the tagged place.
With ELSA, it will be difficult for malicious users to spoof their location and tag a false place or object to mislead other users.

\subsection{Related Work}
Location verification schemes have mainly rely on either the TOA or RSS range-based approaches where the target and anchors are also known as the prover and verifiers respectively.
In range-based approaches, deterministic geometrical boundaries are often used to decide whether to accept or reject localization claims.
Vora \textit{et al.}~\cite{vora2006} adopt a geometric approach to detect location spoofing attacks, using sharply defined boundaries for acceptance (circular zone) and rejection (polygonal zone), with an ambiguity zone between the two boundaries.
Audibility is assumed to be guaranteed within the circular acceptance zone.
Such deterministic methods do not account for the variance of the naturally occurring noise.
Our statistical model, on the other hand, generalizes this approach by accounting for the naturally occurring observation noise via a Gaussian noise term (see (1)) $W_i \sim \mathcal{N}(0, \sigma^2_W)$.
The geometric approach is therefore a special case of our model where $W_i \sim \mathcal{N}(0, 0)$.
Therefore, our stochastic model provides a better representation of the real life wireless conditions  by quantifying the probability of being audible or inaudible and accounts for the naturally occurring observation noise.

In addition, several works used cryptographic security protocols and message exchanges to make it difficult for an attacker to spoof his location.
The work in \cite{hasan2015} presents a framework for using witness nodes to validate the location of targets via a cryptographic asserted location proof protocol to verify their distances to the target.
Next, \cite{fiore2013} presents a similar but distributed cooperative witnesses protocol to verify location claims through a series of message exchanges.
Likewise, \cite{wei2013} proposes a method to check if the target lies within a claimed region and whether the claimed location exceeds a reasonable bound.
Distance bounding protocols (e.g.,~\cite{chiang2012}) have also been proposed to verify that a target is located within a geometric region from the anchors. 
This is achieved by rapidly exchanging of messages based on random nonces to bound the distances between the target and the anchors.

Special features such as 
anonymous beacons are used in \cite{mala2014} to verify a target location.
Capkun \textit{et al.}~\cite{capkun2008} further use hidden and mobile anchors not known by the adversary to verify the location of targets via a 
simple challenge-response scheme.
 Basilico \textit{et al.}~\cite{basilico2014} model the location verification problem as a non-cooperative two-player game between the anchors and the malicious target to compute the best placement for the anchors.
Our work is similar to~\cite{yan2012, yan2014, yan20142, yan2015} which use the information theoretic likelihood ratio test (LRT) approach to verify the location of targets via RSS readings.
We also adopt the LRT framework but tackle the additional challenge of having the anchors localize the target themselves.
Furthermore, we exploit audibility information (which is often ignored) to improve the detection of the location spoofing.

\subsection{Our Contributions}

To the best of our knowledge, this is the first attempt to model and incorporate audibility information to improve location spoofing detection using a statistical approach based on the \emph{missing-not-at-random (MNAR)}~\cite{rubin1976} concept (explained in Section II).
The key contributions of this paper can be summarized as follows:
\begin{itemize}
  \item We introduce the notion of \textit{audibility} and develop a framework for using it to improve the detection of location spoofing attempts.
  \item We design ELSA, an audibility-aware GLRT to detect location spoofing attempts, and prove that it has better detection performances than the conventional non-audibility-aware GLRT.
  \item We verify the efficacy of ELSA using both extensive simulations and a real-world experimental dataset.
\end{itemize}

\subsection{Notation}

Uppercase letters denote random variables and the corresponding lowercase letters their realizations, and bold letters represent vectors.
With a slight abuse of notation, we use lowercase $p(x)$ to represent both the probability density function (pdf) and probability mass function (pmf), and uppercase $P($``event''$)$ to represent the probability of an event.
The normal pdf is represented by $\mathcal{N}(x; \mu, \sigma^2) = \frac{1}{\sigma \sqrt{2\pi}} e^{-\frac{(x-\mu)^2}{2\sigma^2}}$, and the standard normal cumulative distribution function (cdf) by $\Phi(x) = \frac{1}{\sqrt{2\pi}} \int^{x}_{-\infty} e^{-\frac{t^2}{2}} dt$.
Finally, we use $\mathds{1}(\cdot)$ to denote the indicator function which equals one if its argument $(\cdot)$ is true and zero otherwise.

The rest of this paper is organized as follows.
Section II presents a motivating example for our proposed framework.
The analytic model is introduced in Section III and the problem formulation is presented in Section IV.
Section V discusses our experimental results.
Finally, conclusions are drawn in Section VI.

\section{Motivating Example For Proposed Audibility Framework}
\label{sec:motivatingeg}

We first illustrate with an example
the concept of audibility before elaborating an example on how audibility aids in detecting the attacks.

\subsection{How Audibility Aids in Location Spoofing Detection}
Using the conventional trilateration technique \cite{neal2005, misra2009} (without utilizing audibility information), distance estimates from at least three different non-collinear anchors are needed to localize a target.
Otherwise, there may exist ambiguity when there are only two delay measurements.
For example, the target may be equally likely to be at two separate regions as seen from the target's likelihood heat map in Fig.~\ref{fig:toatoyexample}a.
However, this ambiguity can be significantly reduced once we incorporate the audibility information (see Fig.~\ref{fig:toatoyexample}b).
As a result, the bottom right region is now unlikely since there exists a nearby anchor that does not receive any delay measurement (not audible).
Therefore, by taking advantage of the ``missing delay measurements'' or the inaudibility information, we are able to relax the fundamental three distance estimates assumption without using any additional hardware or message exchanges. This leads to an improved accuracy of the TOA localization algorithm at no extra cost.
The audibility information can be exploited because the missing observations are \emph{Missing Not At Random (MNAR)} as termed by Rubin in his seminal work~\cite{rubin1976} where he developed a statistical framework to account for missing data.
Thus, we should not ignore the missing delay observations as it also provides additional information about the target location.

\begin{figure}
\centering \footnotesize
\includegraphics[width=3in]{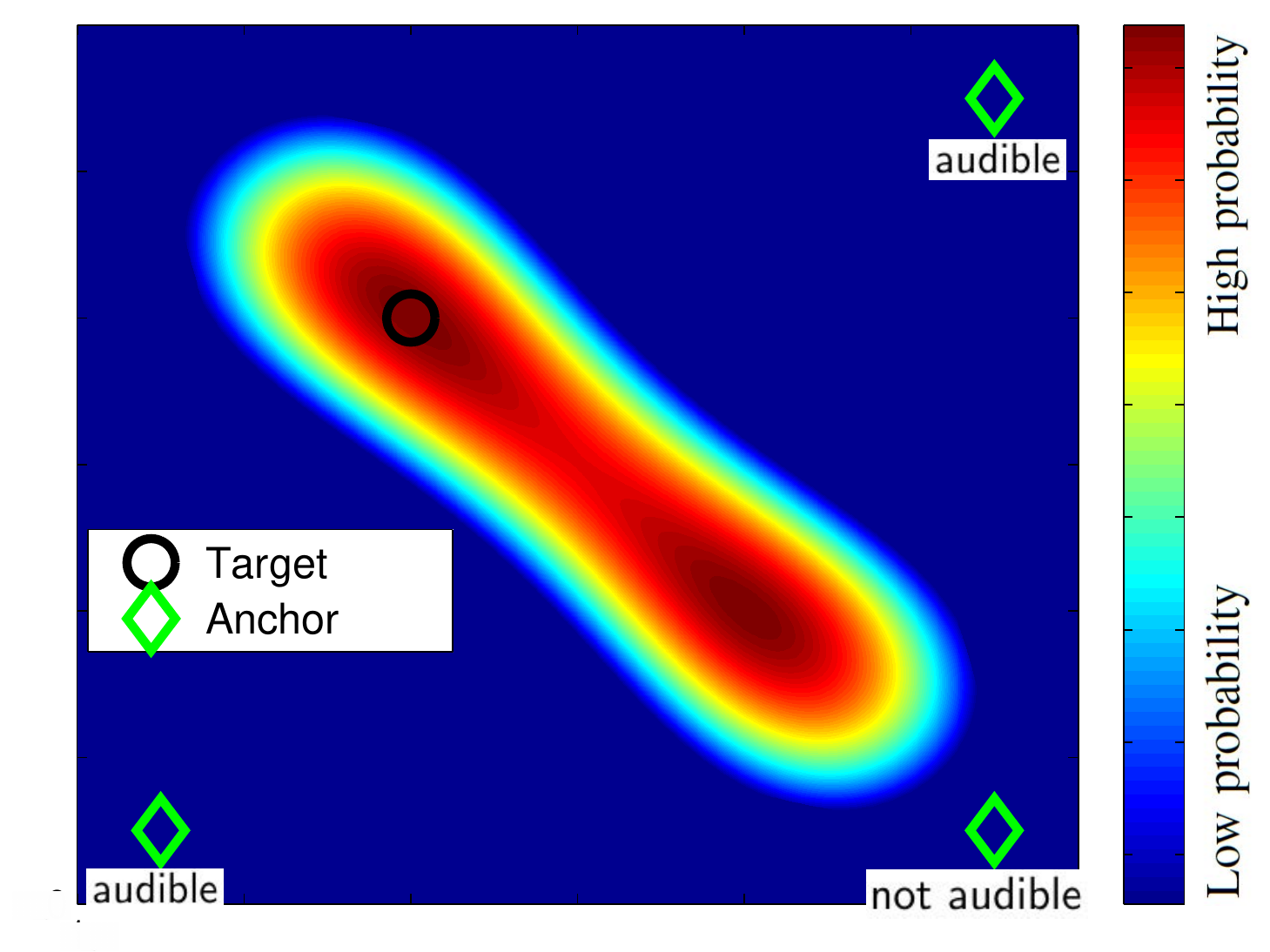}
\\[0.25em]
(a) Conventional TOA likelihood surface.
\\[0.25em]
\includegraphics[width=3in]{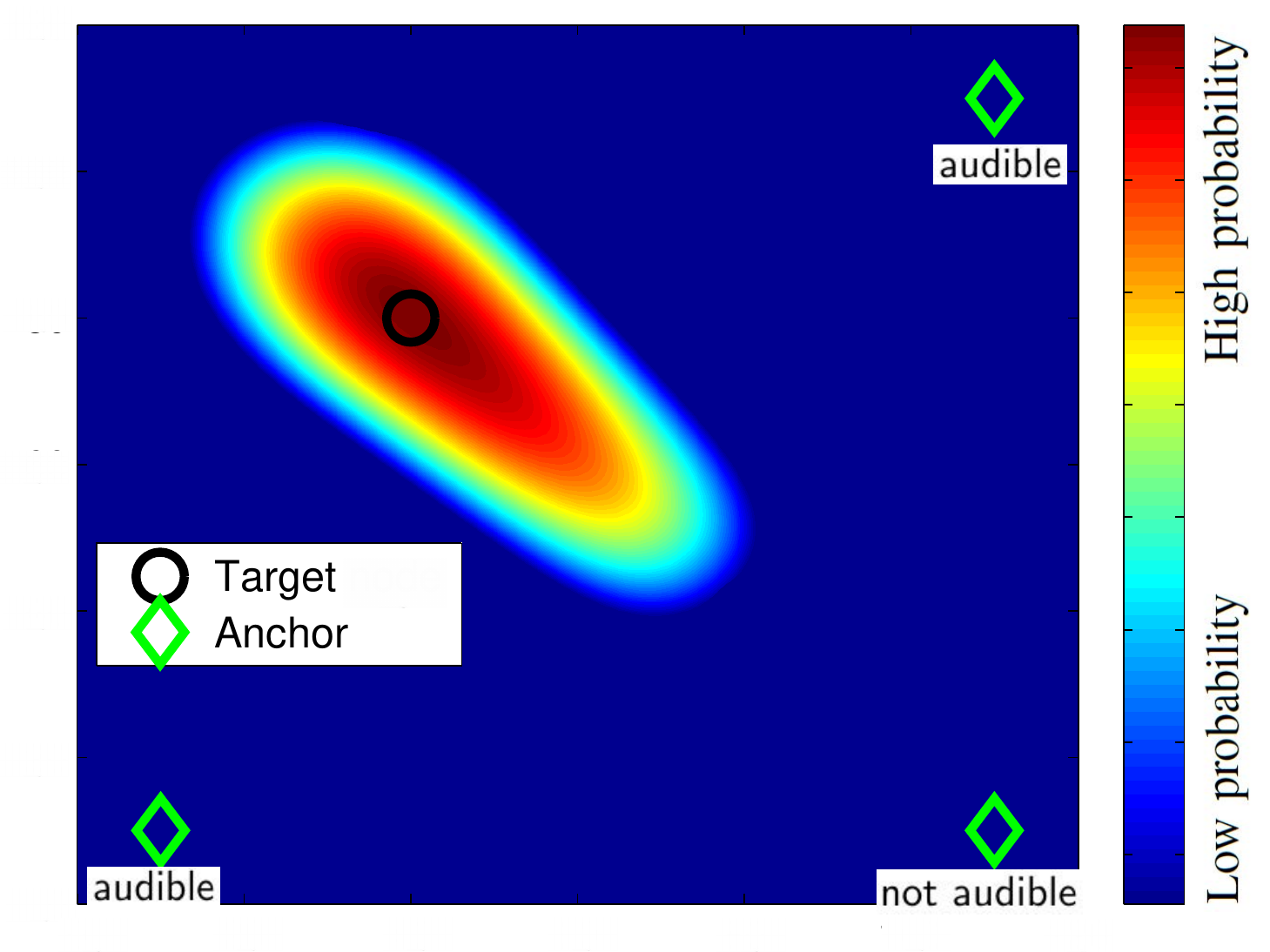}
\\[0.25em]
(b) TOA likelihood surface with audibility information.
\\[0.25em]
\caption{Log-likelihood heat map for the location of a target with three anchors (of which two are audible). Regions with higher probabilities for the target's location are represented by red. }
\label{fig:toatoyexample}
\end{figure}

\subsection{Toy Example on the Use of Audibility Information}
Shown in Fig.~\ref{fig:howitworks} is a room with an anchor at each corner.
Suppose that a malicious target at the left side of the room (denoted by the circle) is in the audible range of two anchors and wishes to spoof its location to appear at the other side of the room (marked with a cross).
If the target is controlled by an adversary, it can add additional delays to increase its TOA delay measurement~\cite{ misra2009, chiang2012, basilico2014, capkun2010, tapo2014} and hence increase the estimated distance from itself to the two anchors. Otherwise, an external adversary may also selectively jam the wireless channel to introduce delays~\cite{abd2014, pav2014, lee2009}.
The threat model will be detailed in Section~\ref{threatmodel}.
Using the conventional approaches, a detection system will not be able to detect the location spoofing attempt as there are insufficient contradictory information to raise suspicions.
However, using the additional implicitly available audibility information as input, it is now unlikely that the target is located at the cross since it is not in the range of the two anchors at the right side of the room.
The target is more likely to be located at the square shown in  Fig.~\ref{fig:howitworks}. (Note that in actual scenarios, the location estimates may be a small region of equally likely points (see Fig.~\ref{fig:toatoyexample}) instead of an exact location point as shown above, but the concept remains the same.)
Hence, we can detect the location spoofing attack by comparing the likelihood probabilities.

\begin{figure}
\centering
\includegraphics[width=3in]{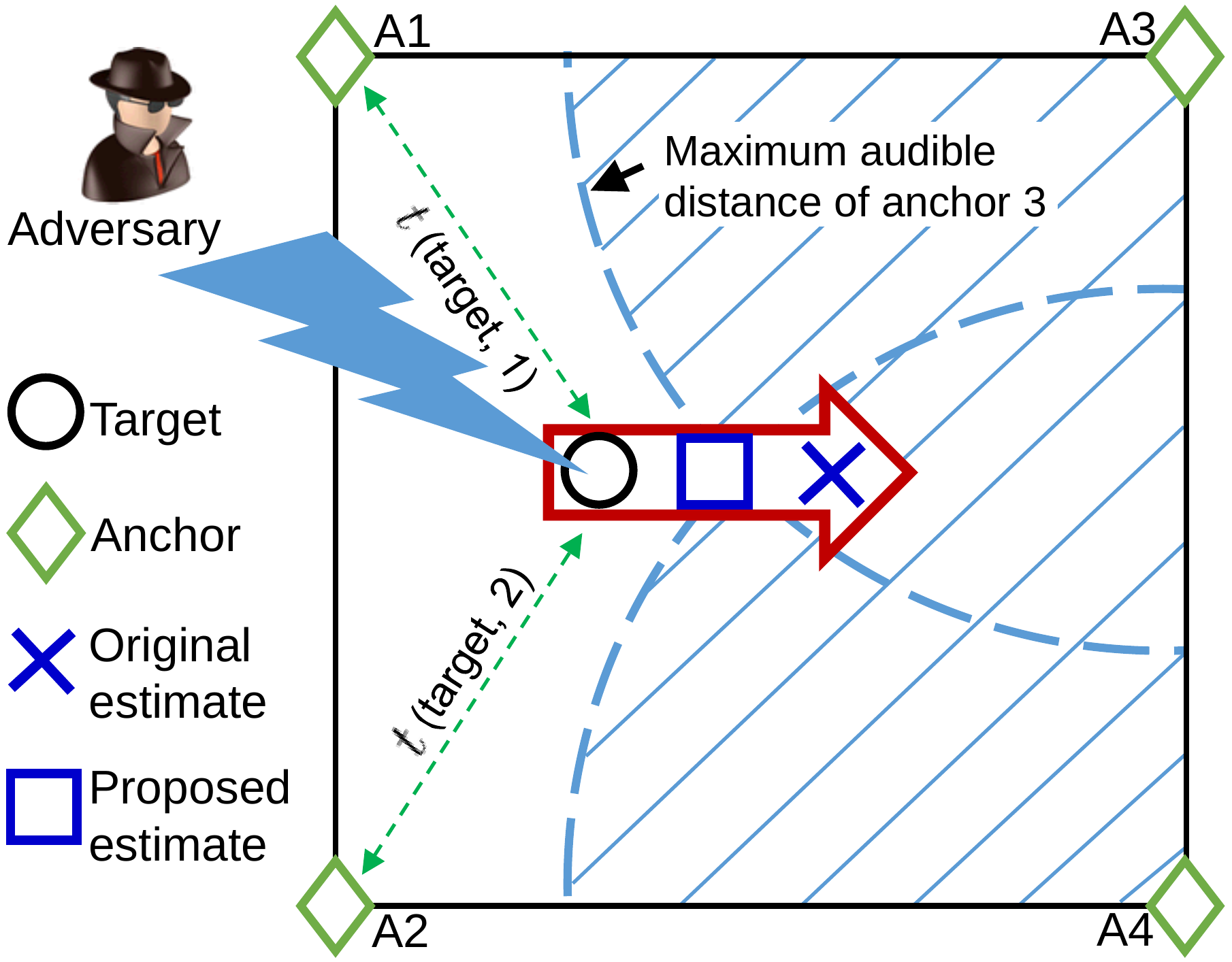}
\caption{Illustration of proposed method where a malicious target attempts to spoof its location by adding delays to the delay measurements $t(\text{target},\text{anchor})$.  }
\label{fig:howitworks}
\end{figure}

\section{Network Model}
In this section, we introduce the definitions for audibility and describe our system and threat models for the location spoofing detection system which uses the TOA-based Two-Way Ranging (TWR) protocol.

\subsection{Connectivity Model}
In order for two nodes A and B to communicate with each other, the transmitted signals should be audible to the other party.
This is modeled as the widely used power loss model~\cite{rappaport2001}.

\begin{definition} [Power loss model]
The received signal power by a node~\textup{A} located at $\mathbf{\Theta}_\textup{A}=\left[x^\textup{(A)}, \;\;y^\textup{(A)}\right]$
from a signal sent by node~\textup{B} which is located at
$\mathbf{\Theta}_\textup{B}=\left[x^\textup{(B)}, \;\;y^\textup{(B)}\right]$ is given by
\begin{align*}
P_R = P_T-10 \alpha \log \frac{d\left(\textup{A},\textup{B}\right)}{d_0}+\epsilon,
\end{align*}
where
$P_T$ is the transmitted power by node~\textup{B}, $\alpha$ is the path-loss exponent,
\newline
$d\left(\textup{A},\textup{B}\right) := \sqrt{\left(x^\textup{(A)}-x^\textup{(B)}\right)^2+\left(y^\textup{(A)}-y^\textup{(B)}\right)^2  }$ is the Euclidean distance between nodes \textup{A} and \textup{B}, $d_0$ is a reference distance and $\epsilon \sim \mathcal{N} \left(0,\sigma_{\epsilon}^2\right)$ represents the shadowing effect.
\end{definition}

If node~B is able to receive signals transmitted by node~A, then the former is said to be audible. More formally, we define audibility as the following.
\begin{definition} [Audibility]
Node~\textup{B} is said to be audible to node~\textup{A} if
\begin{align*}
P_R = P_T-10 \alpha \log \frac{d\left(\textup{A},\textup{B}\right)}{d_0}+\epsilon \geq \lambda,
\end{align*}
where $\lambda$ is a predefined threshold representing the receiver's sensitivity.
\end{definition}

\subsection{Two-Way Ranging (TWR) Distance Estimation Protocol}

\begin{figure}
\centering
\includegraphics[width=3.3in]{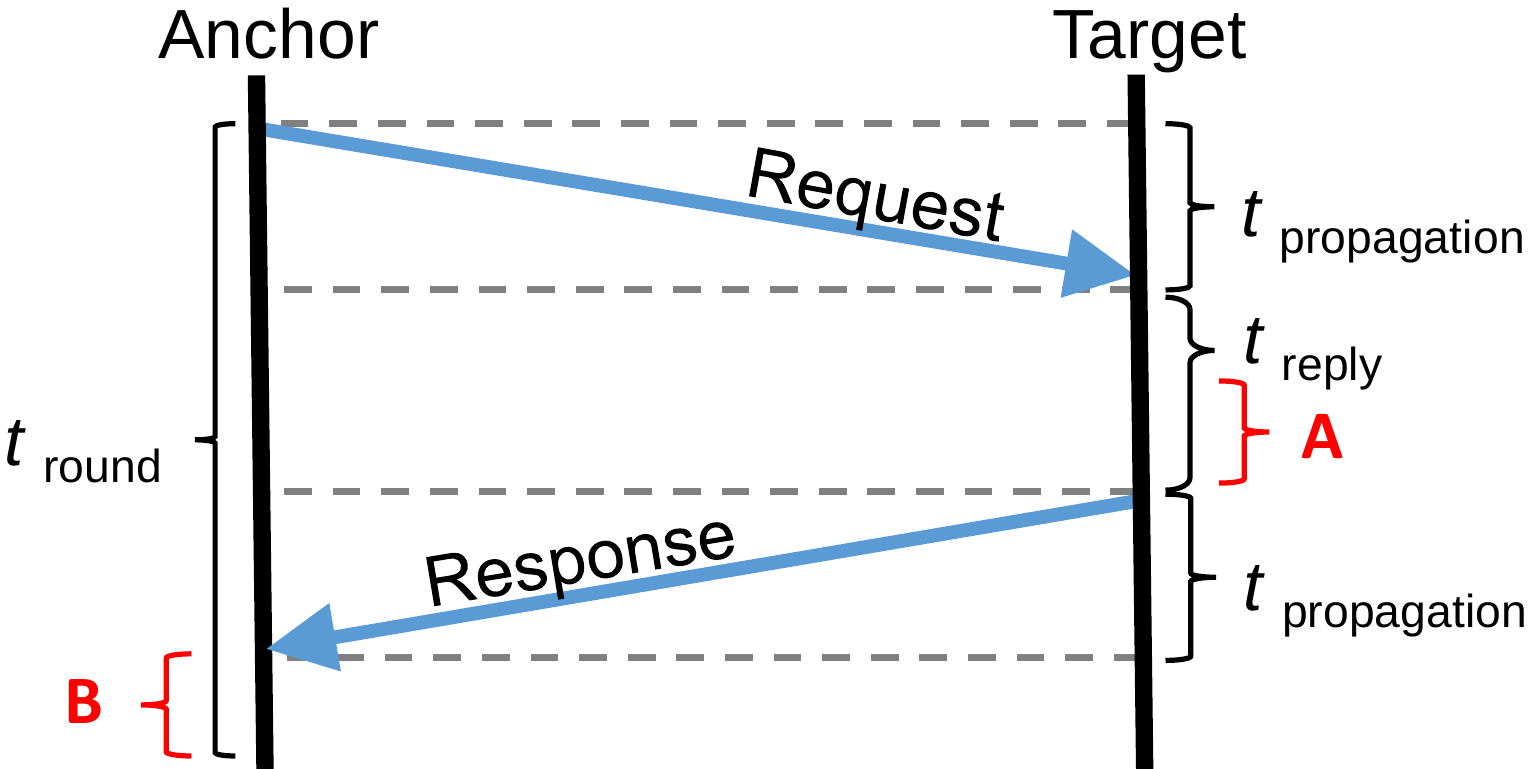}
\caption{Message exchange of the Two-Way Ranging (TWR) distance estimation protocol~\cite{twr}. 
The two arrows marked with letters A and B represent the points of attack by an adversary (see our Threat Model in Section III-D).}
\label{fig:twr}
\end{figure}

The TWR protocol is a time of arrival (TOA)/flight (TOF) based ranging method specified in the IEEE 802.15.4a standard~\cite{twr}. It is gaining popularity especially in small low-cost UWB devices. It allows two communicating devices to estimate their distance from each other without needing time synchronization.
First, the anchor sends a range request packet to an unlocalized target.
The latter then waits for some known time $t_\text{reply}$ before sending a response packet back to the anchor. The value of $t_\text{reply}$ is assumed to be known to both devices.
Assuming that there are no measurement errors, the anchor is able to obtain the round trip time of the two packets $t_\text{round}$ by subtracting the time it first sent a request packet from the time it received the response packet.
Since
	\begin{equation*}
	t_\text{round} = 2\times t_\text{propagation} + t_\text{reply},
	\end{equation*}
the value of the packet propagation \textit{delay} or $t_\text{propagation} = \frac{t_\text{round} - t_\text{reply}}{2}$ can be determined and subsequently the distance between the target and the anchor can be computed as follows:
	\begin{equation*}
	d(\text{target},\text{anchor}) = t_\text{propagation} \times v_p
	\end{equation*}
where $v_p$ is the signal propagation speed.
No time synchronization between the two nodes is required in the TWR protocol as the anchor uses its own local clock information to infer distance. This advantage enables the protocol to be used even with low cost RFID tags where time synchronization is not possible~\cite{lanz2011}. With sufficient range-based distance estimates, a node can be localized using the trilateration or multilateration techniques~\cite{neal2005, misra2009}.

\subsection{System Model}
We consider a scenario where a fusion center receives some delay measurements from its anchors (also known as reference nodes) and transmits the measurements to a backend server for verifying a target's location.
We present the considered wireless system with the following assumptions:

\begin{enumerate}
  \item Assume a wireless network with $n$ static anchors  where the location of the $i^{th}$ anchor (verifier) is denoted by
	\begin{equation*}
	{\bf x}_{i} = [x_{i}, \,\, y_{i}],
	\end{equation*}
where its 2D coordinates $x_{i}, y_{i} \in \mathbb{R}$ for $i\in\{1, \dots, n\}$.

  \item The true location of the target (prover) is denoted by
	\begin{equation*}
	\mathbf{\Theta} = [x_{\theta}, \,\, y_{\theta}],
	\end{equation*}
where its 2D coordinates $x_{\theta}, y_{\theta} \in \mathbb{R}$.
Depending on the deployment scenario, we assume that there is a prior $p(\mathbf{\Theta})$ for the target. A uniform prior can be assigned if the target is equally likely to exist anywhere in the considered region.

  \item We consider a scenario where the TWR protocol~\cite{twr} is used (see Fig.~\ref{fig:twr}). Each anchor $i$ in the communication range of the target will receive a delay measurement~\cite{neal2005} which can be represented by:
	\begin{equation}
	t_i = \frac{d(\mathbf{\Theta}, {\bf x}_i)}{v_p} + W_i,
	\label{eqn:anchornodetoa}
	\end{equation}
	where $d({\bf a}, {\bf b})$ is the Euclidean distance between two locations ${\bf a}, {\bf b}$ and is given by
	\begin{equation}
	d({\bf a}, {\bf b}) = \sqrt{(a_x - b_x)^2 + (a_y - b_y)^2},
	\label{eqn:eucdist}
	\end{equation}
$v_p$ is the signal propagation speed and $W_i$ is the time delay error assumed to be an i.i.d. Gaussian random variable\footnote{Note that $W_i$ may also be come from any other known parametric distribution.} given by $W_i \sim \mathcal{N}(0, \sigma^2_{W})$.

  \item In our audibility model, each anchor $i$ in the communication range of the target will receive a signal with a received power $P_i$ (or received signal strength (RSS)) that is equal or higher than the minimum signal receiving threshold $\lambda$. We use the widely accepted log-normal propagation model~\cite{neal2005} to estimate the received power of the signal:
	\begin{equation}
	P_i = P_{t} - 10 \alpha \log \frac{d(\mathbf{\Theta}, {\bf x}_i)}{d_0} + \epsilon_i \geq \lambda,
	\label{eqn:anchornoderecrss}
	\end{equation}
	where $P_t$ is the received power from the transmitter at a reference distance $d_0$ (typically \SI{1}{m}), $\alpha$ is the path loss exponent,
and $\epsilon_i$ is the received power error assumed to be an i.i.d. Gaussian random variable given by $\epsilon_i \sim \mathcal{N}(0, \sigma^2_{\epsilon})$.

  \item If an anchor $i$ does not receive any signal from the target, we can treat the received signal as having a received power $P_i$ that is less than the minimum signal receiving threshold $\lambda$. i.e.,
$	P_i  < \lambda.$

  \item We let $r_i$ be an indicator variable that depends on whether the anchor $i$ receives a delay measurement from the target (see \eqref{eqn:anchornoderecrss}):
	\begin{equation}
  	r_i = \left\{ \begin{array}{rl}
	1 &\mbox{ if } P_i \geq \lambda, \\
 	0 &\mbox{ otherwise.}
       \end{array} \right.
	\end{equation}

\end{enumerate}

\textit{Empirical Support for Chosen TOA and RSS Models: }
Our chosen TOA and RSS models in \eqref{eqn:anchornodetoa} and \eqref{eqn:anchornoderecrss} respectively are supported by the experimental measurements obtained from \cite{patwari2003}. The TOA and RSS measurements are plotted in Figs.~\ref{fig:toa_line_fit} and \ref{fig:rss_line_fit} respectively.
As seen from the figures, the zero-mean Gaussian noise and linearity assumptions (see ``Robust Fit'', a MATLAB function which uses reweighted least squares) are reasonable and provide good representation of the actual data.
A Kolmogorov-Smirov (KS) test was also used in \cite{patwari2003} which showed that the Gaussian assumption is valid under a 0.05 significance level.

\begin{figure}[!ht]
  \centering
  \includegraphics[width=3in]{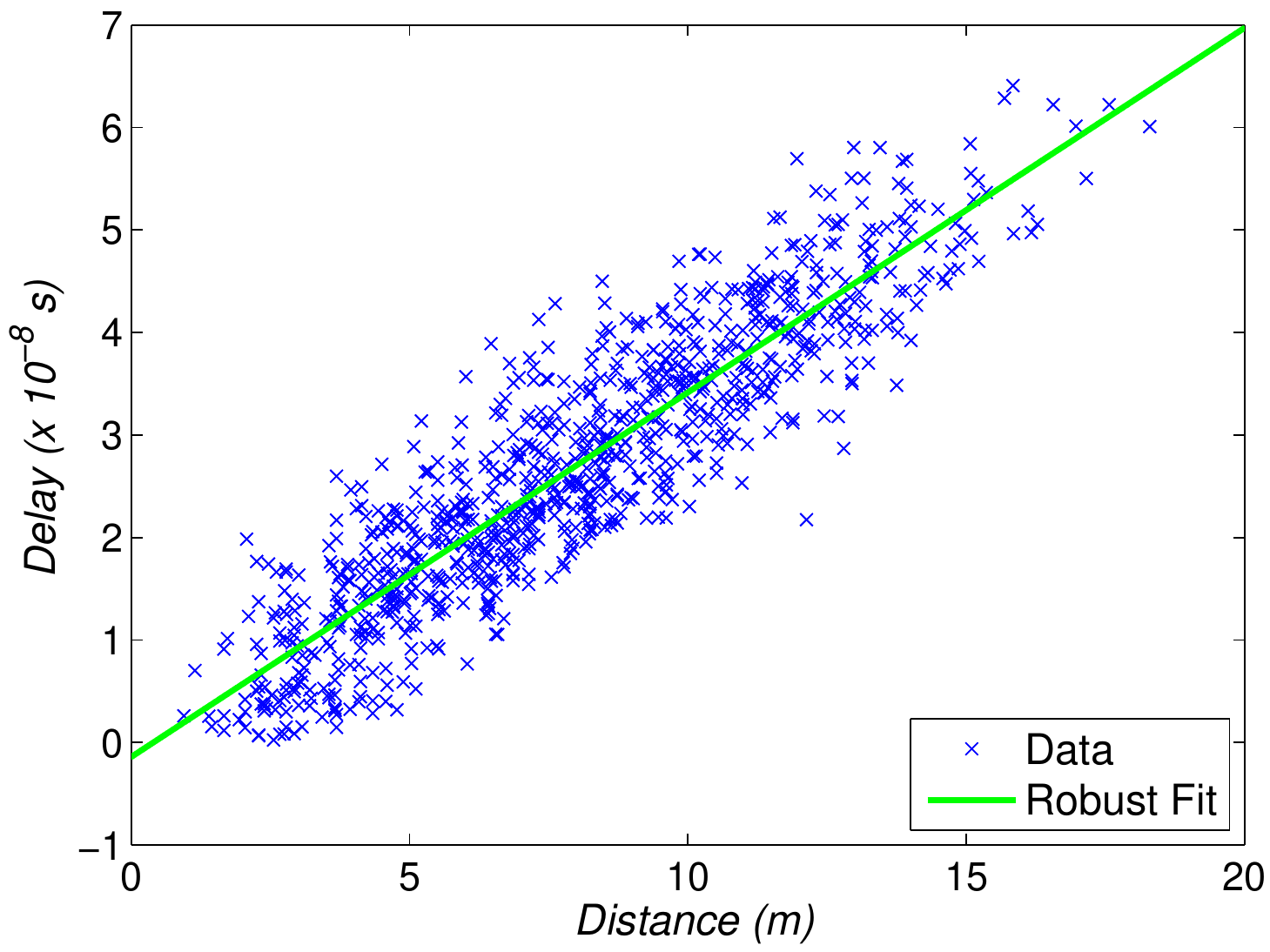}
  \caption{TOA delay (from empirical data~\cite{patwari2003}) as a function of distance between two nodes.
A Kolmogorov-Smirov (KS) test was used in\mbox{\cite{patwari2003}} which showed that the Gaussian assumption is valid under a 0.05 significance level.}
  \label{fig:toa_line_fit}
\end{figure}

\begin{figure}[!h]
  \centering
  \includegraphics[width=3in]{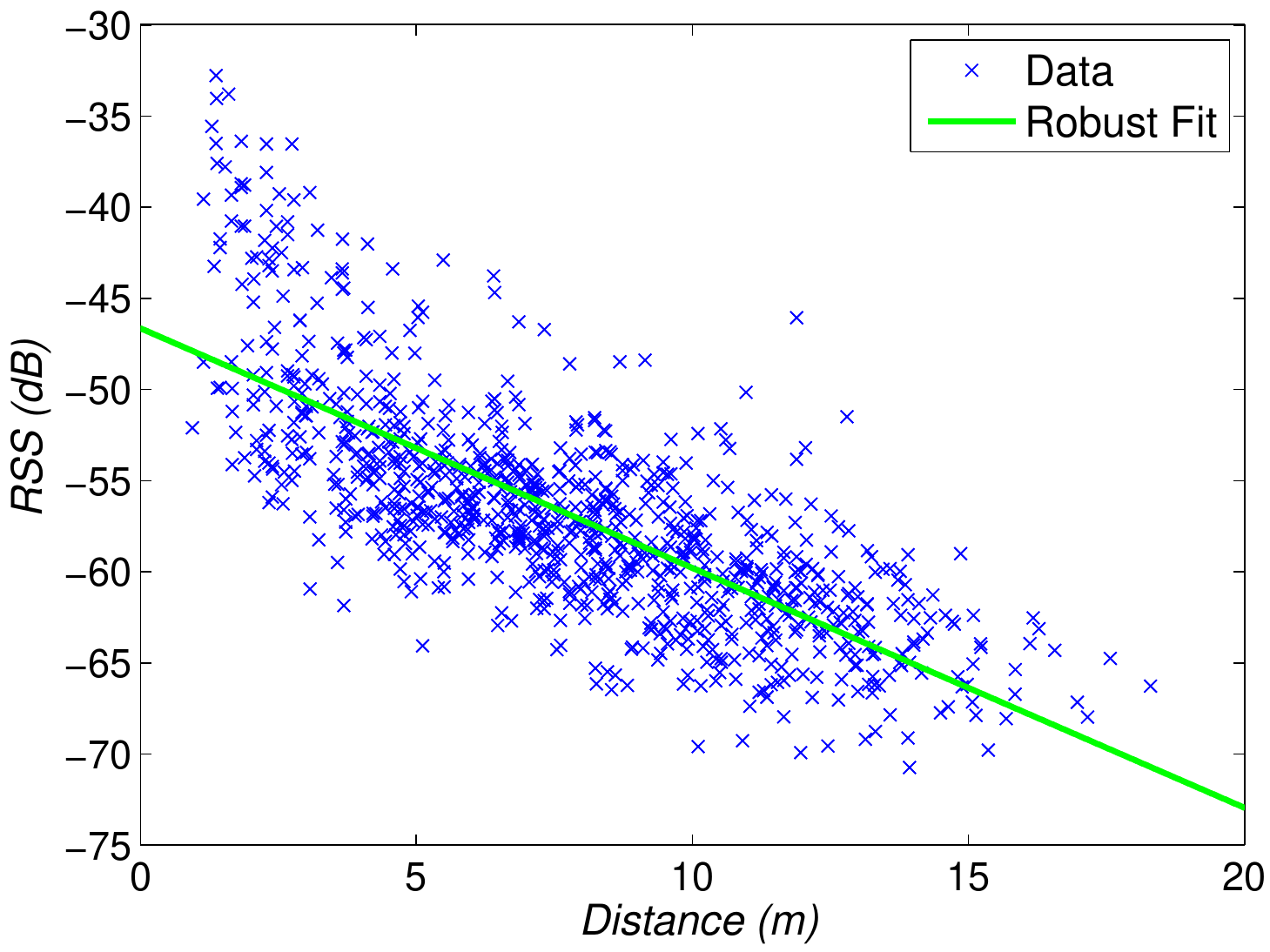}
  \caption{RSS value (from empirical data~\cite{patwari2003}) as a function of distance between two nodes.
A Kolmogorov-Smirov (KS) test was used in\mbox{\cite{patwari2003}} which showed that the Gaussian assumption is valid under a 0.05 significance level.}
  \label{fig:rss_line_fit}
\end{figure}

\subsection{Threat Model}
\label{threatmodel}
We consider an internal and external adversary
 whose main goal is to significantly perturb a target's perceived location by the fusion center $\mathbf{\widehat{\Theta}}$ from its true location  $\mathbf{\Theta}$
 by \textit{manipulating the response time of the target, thus affecting the TOA delay measurements}
received by the anchors
as discussed in our motivating example in Section~\ref{sec:motivatingeg}.
Recall from Section III-C that the delay measurement received at the $i^{th}$ anchor in a non-adversarial environment is given by:
	\begin{equation*}
	t_i = \frac{d(\mathbf{\Theta}, {\bf x}_i)}{v_p} + W_i.
	\end{equation*}
where $W_i$ represents the i.i.d. Gaussian random variable time delay error.
A malicious target or anchor (in the case of an \textit{internal} adversary) can falsify the target's location by adding a delay $\delta_i$ before replying a TWR request message such that the received delay measurement becomes
	\begin{equation}
	t_i = \frac{d(\mathbf{\Theta}, {\bf x}_i)}{v_p} + W_i + \delta_i
  \label{eqn:attackedt}
	\end{equation}
where we assume $\delta_i \stackrel{\text{i.i.d.}}{\sim}  \mathcal{N}(\mu_{\delta}, \sigma^2_{\delta})$.
The malicious target can insert the delay at point A (as shown in Fig.~\ref{fig:twr}) in the TWR protocol while a malicious anchor may insert the delay at point B. The malicious targets may collude to fool the anchors by appearing to be closer or further from them.
This scenario is accounted by the i.i.d. Gaussian noise model in~\eqref{eqn:attackedt}.
A positive attacker delay will fool an anchor into believing that the target is further away from its actual position while a negative delay will make the target appear nearer to the anchor than it really is.
The collusion is assumed to be limited to nearby nodes in the vicinity due to our Gaussian attacker delay model.
The Gaussian model is used for analytical convenience as the adversary may be able to launch 
distance enlargement or distance reduction attacks \cite{pot2010}.

Since the distance estimate computed by an anchor $i$ is equivalent to
	\begin{equation}
	\begin{split}
	\widehat{d}(\mathbf{\Theta}, {\bf x}_i) &= t_i  v_p
= \left( \frac{d(\mathbf{\Theta}, {\bf x}_i)}{v_p} + W + \delta_i \right)  v_p,
	\end{split}
	\end{equation}
where $v_p \gg 0$, depending on the carrier frequency,
a small value of delay $\delta_i$ (e.g., \mbox{$10^{-9}$\,s}) is sufficient to {  result in a large difference in the estimated distance
(approximately 12.5cm, in the case of 2.4 GHz radio waves).
An \textit{external} adversary (who cannot compromise nodes) may also increase the delay measurement by some $\delta_i$, which may not necessarily be non-negative
through  attacking the PHY layer\cite{tapo2014}.

Our goal is to design a detection test that runs at the backend server, independent of the protocols used between the anchors and targets (see Fig.~{\ref{fig:system_model}}) for data communications, authentication, network registration, etc.
This is to maximize compatibility with existing legacy TWR systems.
Therefore, the proposed test is flexible enough to be used in both scenarios with low power low-computational power IoT devices and scenarios with high computational power IoT devices.

 	%\clearpage
\section{ELSA: Enhanced Location Spoofing Detection using Audibility}
We present the \emph{location spoofing} detection test ELSA which utilizes both TOA measurements and the implicitly available \emph{audibility} information to verify that a target is not spoofing its delay measurements.

\subsection{Problem Formulation: Optimal Detection}
In order to achieve this task, a common approach would be to construct a binary hypothesis test to verify the received delay measurements.
The well-known Likelihood Ratio Test (LRT) which is the optimal test (justified by the Neyman-Pearson lemma~\cite{neyman1933, kay1993}) can be used to detect location spoofing attempts under the two competing hypotheses:
	\begin{equation}
	\begin{split}
	&\mathcal{H}_0: \text{no location spoofing}
	\\
	&\mathcal{H}_1: \text{location spoofing attempt}.
	\end{split}
	\end{equation}
The LRT\footnote{The LRT for the conventional non-audibility-aware approach (see Appendix-\ref{Appendix_woaudibility1}) is $
	\Lambda({\bf t})
	\triangleq                                                                                                   \frac{ p({\bf t}| \mathcal{H}_1) }{ p({\bf t} | \mathcal{H}_0) }
	\underset{\mathcal{H}_0}{\overset{\mathcal{H}_1}{\gtrless}}
 	\eta.$ }
 	can be formulated as:
	\begin{equation}
	\Lambda({\bf t}, {\bf r})
	\triangleq                                                                                                   	\frac{ p({\bf t}, {\bf r} | \mathcal{H}_1) }{ p({\bf t}, {\bf r} | \mathcal{H}_0) }
	\underset{\mathcal{H}_0}{\overset{\mathcal{H}_1}{\gtrless}}
 	\eta,
	\label{eqn:lrt}
	\end{equation}
where the bold letters ${\bf t}$ and ${\bf r}$ represent vectors of delay observations ${\bf t} = [t_1, \dots, t_n]$ and audibility indicator values ${\bf r} = [r_1, \dots, r_n]$ from the $n$ anchors respectively, and $\eta$ is a chosen threshold.

Under the Neyman-Pearson lemma,  the LRT is the most powerful test at each significance level $\alpha$ (false alarm) for a threshold $\eta$ where $p(\Lambda({\bf t}, {\bf r}) > \eta | \mathcal{H}_0) = \alpha$.
The functions $p({\bf t}, {\bf r} | \mathcal{H}_0)$ and $p({\bf t}, {\bf r} | \mathcal{H}_1)$ represent the likelihood functions for the null hypothesis and alternative hypothesis respectively.
Since we treat $\mathbf{\Theta}$ as an unknown random variable,
the likelihood functions can be formulated as
	\begin{equation*}
	\begin{split}
	p({\bf t}, {\bf r} | \mathcal{H}_j)
	&= \int p({\bf t}, {\bf r} | \mathbf{\Theta}, \mathcal{H}_j) p(\mathbf{\Theta}|\mathcal{H}_j)  \,d\mathbf{\Theta}
	\\
	&= \int p({\bf t} | {\bf r}, \mathbf{\Theta}, \mathcal{H}_j) p({\bf r} | \mathbf{\Theta}, \mathcal{H}_j) p(\mathbf{\Theta} | \mathcal{H}_j)  \,d\mathbf{\Theta}.
	\end{split}
	\end{equation*}

However, a closed form expression to the above integral is intractable due to the non-linear relationship in $p({\bf t}, {\bf r} | \mathbf{\Theta}, \mathcal{H}_j)$.
Hence, the LRT in \eqref{eqn:lrt} is no longer applicable.
Instead, it is common to use the Generalized Likelihood Ratio Test (GLRT)~\cite{neyman1933, kay1993}), given by
	\begin{equation}
	\Lambda({\bf t}, {\bf r})
	\triangleq                                                                                                   \frac{p( {\bf t}, {\bf r} |\mathcal{H}_1, \widehat{\mathbf{\Theta}}_{\MAP}^{\mathcal{H}_1})}     {p( {\bf t}, {\bf r} |\mathcal{H}_0, \widehat{\mathbf{\Theta}}_{\MAP}^{\mathcal{H}_0})}                                      \underset{\mathcal{H}_0}{\overset{\mathcal{H}_1}{\gtrless}}
 \eta,
	\label{eqn:glrt}
	\end{equation}
where we approximate $p( {\bf t}, {\bf r} |\mathcal{H}_j)$ using the maximum-a-posteriori (MAP) estimate $\widehat{\mathbf{\Theta}}_{\MAP}^{\mathcal{H}_j}$ (see Appendix-B) given by
	\begin{align*}
\argmax_{\mathbf{\Theta}}  &\sum_{i=1}^{n}
	 \log \mathcal{N}(t_i;  \frac{d(\mathbf{\Theta}, {\bf x}_i)}{v_p} +\delta_i, \sigma^2_{W})
 \mathds{1}(r_i = 1)
	\\&
	 \,\,\,\,\,\, + \sum_{i=1}^{n}
\log P(r_i = 1 | \mathbf{\Theta}, \mathcal{H}_j) \mathds{1}(r_i = 1)
	\\&
	 \quad
+ P(r_i = 0 | \mathbf{\Theta}, \mathcal{H}_j) \mathds{1}(r_i = 0)
+
\log p(\mathbf{\Theta} | \mathcal{H}_j).
	\end{align*}

\subsection{Derivation of Test Statistic}
Under the null hypothesis $\mathcal{H}_0$, the likelihood function is simply
	\begin{equation}
	p( {\bf t}, {\bf r}| \mathcal{H}_0, \widehat{\mathbf{\Theta}}_{\MAP}^{\mathcal{H}_0} ) =
p( {\bf t}| {\bf r}, \mathcal{H}_0, \widehat{\mathbf{\Theta}}_{\MAP}^{\mathcal{H}_0})
p( {\bf r}| \mathcal{H}_0, \widehat{\mathbf{\Theta}}_{\MAP}^{\mathcal{H}_0}),
	\label{eqn:nullhypothesis}
	\end{equation}
where
	\begin{equation*}
	\begin{split}
	p( {\bf t}| {\bf r}, \mathcal{H}_0, \widehat{\mathbf{\Theta}}_{\MAP}^{\mathcal{H}_0})
	=
\prod^{n}_{i = 1}
\Big[
 \mathcal{N}(t_i;  \frac{d(\widehat{\mathbf{\Theta}}_{\MAP}^{\mathcal{H}_0}, {\bf x}_i)}{v_p}, \sigma^2_{W})\mathds{1}(r_i = 1)
	+ \mathds{1}(r_i = 0)
\Big],
	\end{split}
	\end{equation*}
and
	\begin{equation*}
	\begin{split}
	p( {\bf r}| \mathcal{H}_0, \widehat{\mathbf{\Theta}}_{\MAP}^{\mathcal{H}_0}) =
  \prod^{n}_{i = 1}
\Big[ P(r_i = 1 | \widehat{\mathbf{\Theta}}_{\MAP}^{\mathcal{H}_0}) \mathds{1}(r_i = 1)
	+ 	P(r_i = 0 | \widehat{\mathbf{\Theta}}_{\MAP}^{\mathcal{H}_0}) \mathds{1}(r_i = 0)
\Big].
	\end{split}
	\end{equation*}

Under the alternative hypothesis $\mathcal{H}_1$,
	\begin{equation}
	p( {\bf t}, {\bf r}| \mathcal{H}_1, \widehat{\mathbf{\Theta}}_{\MAP}^{\mathcal{H}_1} ) =
p( {\bf t}| {\bf r}, \mathcal{H}_1, \widehat{\mathbf{\Theta}}_{\MAP}^{\mathcal{H}_1})
p( {\bf r}| \mathcal{H}_1, \widehat{\mathbf{\Theta}}_{\MAP}^{\mathcal{H}_1}),
	\label{eqn:althypothesis}
	\end{equation}
where
	\begin{equation*}
	\begin{split}
	p( {\bf t}| {\bf r}, \mathcal{H}_1, \widehat{\mathbf{\Theta}}_{\MAP}^{\mathcal{H}_1}) 	 =
   \prod^{n}_{i = 1}
\Big[
\mathcal{N}(t_i;  \frac{d(\widehat{\mathbf{\Theta}}_{\MAP}^{\mathcal{H}_1}, {\bf x}_i)}{v_p} + \mu_{\delta}, \sigma^2_{W} + \sigma^2_{\delta})\mathds{1}(r_i = 1) + \mathds{1}(r_i = 0)
\Big],
	\end{split}
	\end{equation*}
and
	\begin{equation*}
	\begin{split}
	p( {\bf r}| \mathcal{H}_1, \widehat{\mathbf{\Theta}}_{\MAP}^{\mathcal{H}_1})  =
 \prod^{n}_{i = 1}  \left[ P(r_i = 1 | \widehat{\mathbf{\Theta}}_{\MAP}^{\mathcal{H}_1}) \mathds{1}(r_i = 1)
	+ 	P(r_i = 0 | \widehat{\mathbf{\Theta}}_{\MAP}^{\mathcal{H}_1}) \mathds{1}(r_i = 0) \right].
	\end{split}
	\end{equation*}

Substitution of the values obtain from \eqref{eqn:nullhypothesis} and \eqref{eqn:althypothesis} into \eqref{eqn:glrt} will give the test statistic in \eqref{eqn:glrtaudibiltiy}.
\begin{figure*}
	\begin{equation}
	\begin{split}
		 &\Lambda({\bf t}, {\bf r}) =
  \prod^{n}_{i = 1}
\left[
\mathcal{N}(t_i;  \frac{d(\widehat{\mathbf{\Theta}}_{\MAP}^{\mathcal{H}_1}, {\bf x}_i)}{v_p} + \mu_{\delta}, \sigma^2_{W} + \sigma^2_{\delta})\mathds{1}(r_i = 1) + \mathds{1}(r_i = 0)
\right]
	\\& \times
\prod^{n}_{i = 1}  \left[ P(r_i = 1 | \widehat{\mathbf{\Theta}}_{\MAP}^{\mathcal{H}_1}) \mathds{1}(r_i = 1)
	+ 	P(r_i = 0 | \widehat{\mathbf{\Theta}}_{\MAP}^{\mathcal{H}_1}) \mathds{1}(r_i = 0) \right]
	/
    \left[
\prod^{n}_{i = 1} \mathcal{N}(t_i;  \frac{d(\widehat{\mathbf{\Theta}}_{\MAP}^{\mathcal{H}_0}, {\bf x}_i)}{v_p}, \sigma^2_{W})\mathds{1}(r_i = 1) + \mathds{1}(r_i = 0)
\right]
	\\& \times
 \left[ P(r_i = 1 | \widehat{\mathbf{\Theta}}_{\MAP}^{\mathcal{H}_0}) \mathds{1}(r_i = 1)
	+ 	P(r_i = 0 | \widehat{\mathbf{\Theta}}_{\MAP}^{\mathcal{H}_0}) \mathds{1}(r_i = 0) \right]
	\underset{\mathcal{H}_0}{\overset{\mathcal{H}_1}{\gtrless}}
 	\eta.
	\label{eqn:glrtaudibiltiy}
	\end{split}
	\end{equation}
	\end{figure*}
Algorithm 1 summarizes the steps in the proposed ELSA.
\begin{algorithm}
    \SetKwInOut{Input}{Input}
    \SetKwInOut{Output}{Output}
    \underline{function ELSA$(t_{1,\ldots,n}, x_{1,\ldots,n}, \eta, v_p, d_0, P_t, \lambda, \mu_{\delta}, \alpha, \sigma^2_W, \sigma^2_{\epsilon},  \sigma^2_{\delta} )$}\;
    \Input{Delay measurements received from the target $t_{1,\ldots,n}$, positions of the anchors $x_{1,\ldots,n}$, threshold $\eta$, and the system parameters.}
    \Output{Binary result of hypothesis test.}
    Compute MAP estimate for $\mathcal{H}_0$ (no location spoofing), $\widehat{\mathbf{\Theta}}_{\MAP}^{\mathcal{H}_0}$ via \eqref{eqn:mapestimate} (with $\delta_i = 0$).
\\
    Compute MAP estimate for
	$\mathcal{H}_1$ (location spoofing attempt), $\widehat{\mathbf{\Theta}}_{\MAP}^{\mathcal{H}_1}$ via \eqref{eqn:mapestimate} (with $\delta_i \neq 0$).
\\
	Compute likelihood probabilities for the two MAP estimates, $p( {\bf t}, {\bf r} |\mathcal{H}_0, \widehat{\mathbf{\Theta}}_{\MAP}^{\mathcal{H}_0} )$ and
$p( {\bf t}, {\bf r} |\mathcal{H}_1, \widehat{\mathbf{\Theta}}_{\MAP}^{\mathcal{H}_1} )$
 via \eqref{eqn:nullhypothesis} and \eqref{eqn:althypothesis} respectively.
\\
	Compute the decision rule $\Lambda({\bf t}, {\bf r}) = \frac{p( {\bf t}, {\bf r} |\mathcal{H}_1, \widehat{\mathbf{\Theta}}_{\MAP}^{\mathcal{H}_1})}     {p( {\bf t}, {\bf r} |\mathcal{H}_0, \widehat{\mathbf{\Theta}}_{\MAP}^{\mathcal{H}_0})}                                     $ via \eqref{eqn:glrtaudibiltiy}.
\\
	Reject $\mathcal{H}_0$ (no location spoofing) if $\Lambda({\bf t}, {\bf r}) > \eta$. Otherwise, accept $\mathcal{H}_1$  (location spoofing detected).
    \caption{ELSA algorithm for detecting location spoofing attempts.}
\end{algorithm}

Next, we prove using the following theorem that ELSA provides better detection rates than the conventional non-audibility-aware GLRT for the same false alarm rate tradeoff.

\begin{theorem}
For a fixed false alarm rate, the proposed audibility-aware GLRT will have a detection rate $P_d^{\rm A}$ that is higher than the conventional GLRT  $P_d^{\rm NA}$ which does not take into account audibility. i.e.,
	\begin{equation*}
	P_d^{\rm A}  \geq P_d^{\rm NA}.
	\end{equation*}
\end{theorem}

\begin{proof}
See Appendix-{\ref{Appendix_theorem}}.
\end{proof}

\section{Experimental Results and Discussion}
In this section, we evaluate the performance of our proposed detection test ELSA against the conventional (labeled as `original' in the figures) GLRT which does not take into account audibility (similar to the work in \cite{yan20142}) in terms of the location spoofing detection performance.
Both simulations and data from a real-world dataset (available in \cite{datasetweb}) were used in our evaluation.
The MATLAB code used to obtain the simulation results is available as supplemental material, which can be found at \cite{ourmatlabcode}.
Unless otherwise stated, the parameters in Table~\ref{table:simpara} were used in our simulations.

\begin{table}
\center
\caption{Simulation Parameters}
\label{table:simpara}
\vspace{0.5em}
\begin{tabular}{ll}
\toprule
Parameter              								& Value (equivalent distance)  		       \\
\midrule
TOA noise, $\sigma_W$       					& \mbox{$10^{-8}$\,s}  		 \quad\quad (\SI{3}{m})  \\
RSS noise, $\sigma_{\epsilon}$    			& \mbox{$\sqrt{10}$\,dBm} 		           \\
Attacker's delay mean, $\mu_{\delta}$        		& \SI{4e-8}{s} (\SI{12}{m})   \\
Attacker's delay s.d., $\sigma_{\delta}$            		& \SI{4e-8}{s} (\SI{12}{m})     \\
$^*$Only positive attacker delays $|\delta_i|$ were used. & (See Equation~\eqref{eqn:attackedt}) \\
Path loss exponent, $\alpha$            		& 3.2         				\\
Transmit power, $P_t$ at $d_0 = $ \SI{1}{m}		& $-$ \SI{40}{dBm}   				\\
Signal receiving threshold, $\lambda$		& $-$ \SI{102}{dBm} 				 \\
\bottomrule
\end{tabular}
\end{table}

\subsection{Simulation Results for Synthetic Data}
\label{sec:sim_synthetic}
First, we study the effects of utilizing the audibility information using simulations.
We consider the scenario where there exist three anchors at the corners of a \SI{100}{m} $\times$ \SI{100}{m} area as shown in Fig.~\ref{fig:toatoyexample} and the target is selected uniformly at random inside this area (hence, $p(\mathbf{\Theta}) = \frac{1}{100} \times \frac{1}{100}$). We used a grid search with a one meter granularity to search for the optimal target location using the MAP approach (see \eqref{eqn:mapestimate}).
A finer granularity would improve the accuracy of the schemes, but the improvement would not be significant.
Under an adversarial environment, the received delay measurements are adjusted accordingly as discussed in our threat model in Section~\ref{threatmodel}.

\begin{figure}[t]
\centering
\includegraphics[width=3in]{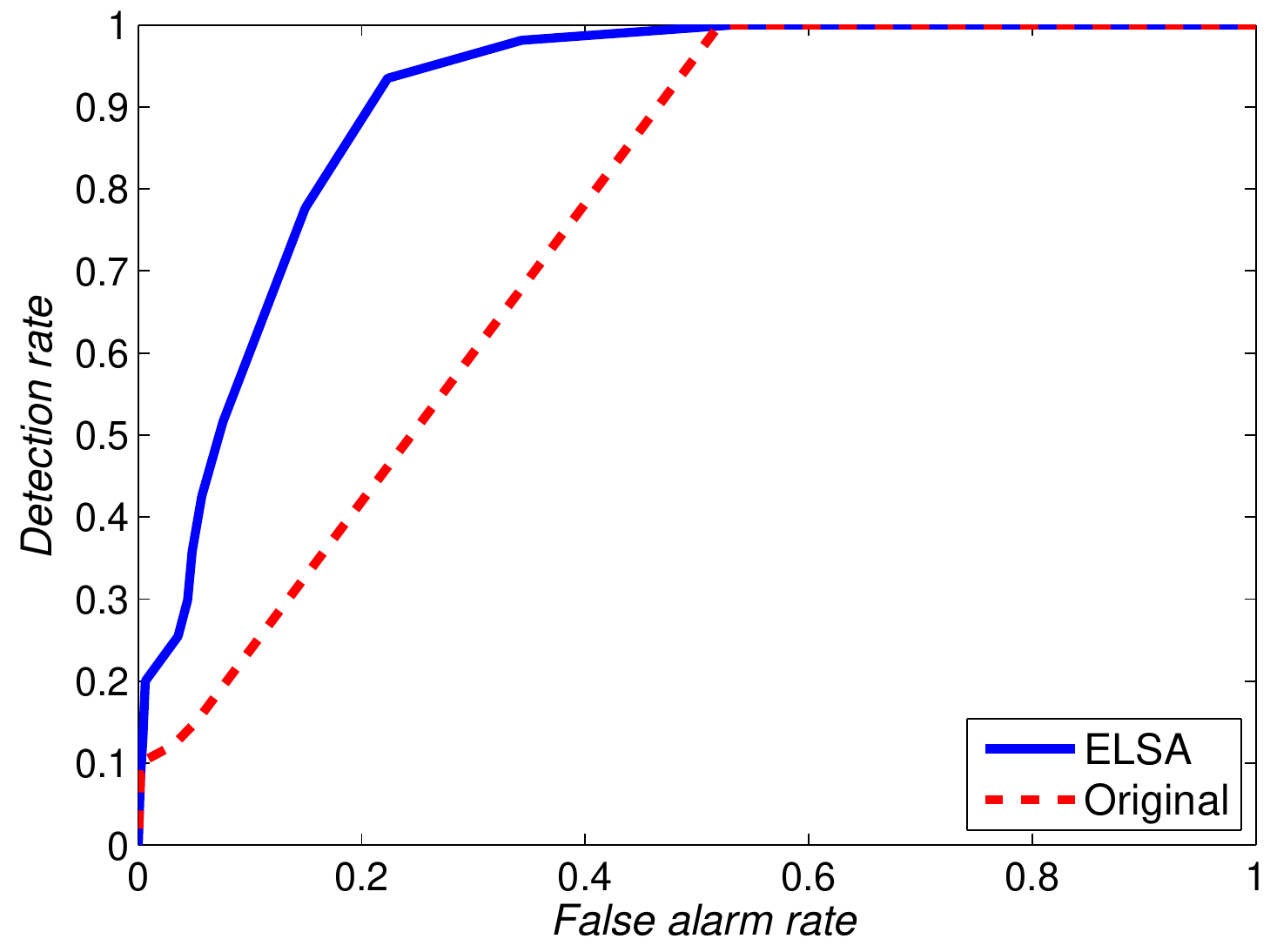}
\caption{ROC curves for 3 anchors (of which 2 are audible).}
\label{fig:toa_roc}
\end{figure}

\subsubsection{\textit{ROC Curve Performance}}
We use the Receiver Operating Characteristic (ROC) curve to compare the detection and false alarm performances of ELSA, against the conventional non-audibility-aware GLRT.
For a given decision rule $\eta$, the detection rate is given by
	\begin{equation*}
	P( \Lambda({\bf t}, {\bf r}) >\eta  |  \mathcal{H}_1  ),
	\end{equation*}
and the false alarm rate is given by
	\begin{equation*}
	P( \Lambda({\bf t}, {\bf r}) >\eta  |  \mathcal{H}_0  ).
	\end{equation*}

In Fig.~\ref{fig:toa_roc}, we plot the ROC curves for scenarios when an attacker adds a positive delay to the delay measurements received by the anchors and the target is on the range of exactly two audible anchors.
The ROC curve for ELSA indicates a significantly better detection performance which demonstrates the superiority of our approach.
Despite a slight model mismatch, an attacker who only adds positive delays does not significantly degrade the detection rate of ELSA.
The detection performance of the conventional approach however, is lower than ELSA's as it is difficult to detect the attack without making use of additional information from the third anchor. Despite not receiving any observations from the third anchor, this piece of valuable information itself is exploited by ELSA whereas the conventional approach simply ignores this.
As it is unlikely that the attacker is able to reduce the propagation delay of a radio wave signal, we only used a positive attacker delay (considered by most works in the literature~\cite{ misra2009, chiang2012, basilico2014}) in our comparisons.

\begin{figure}[t]
\centering
\includegraphics[width=3in]{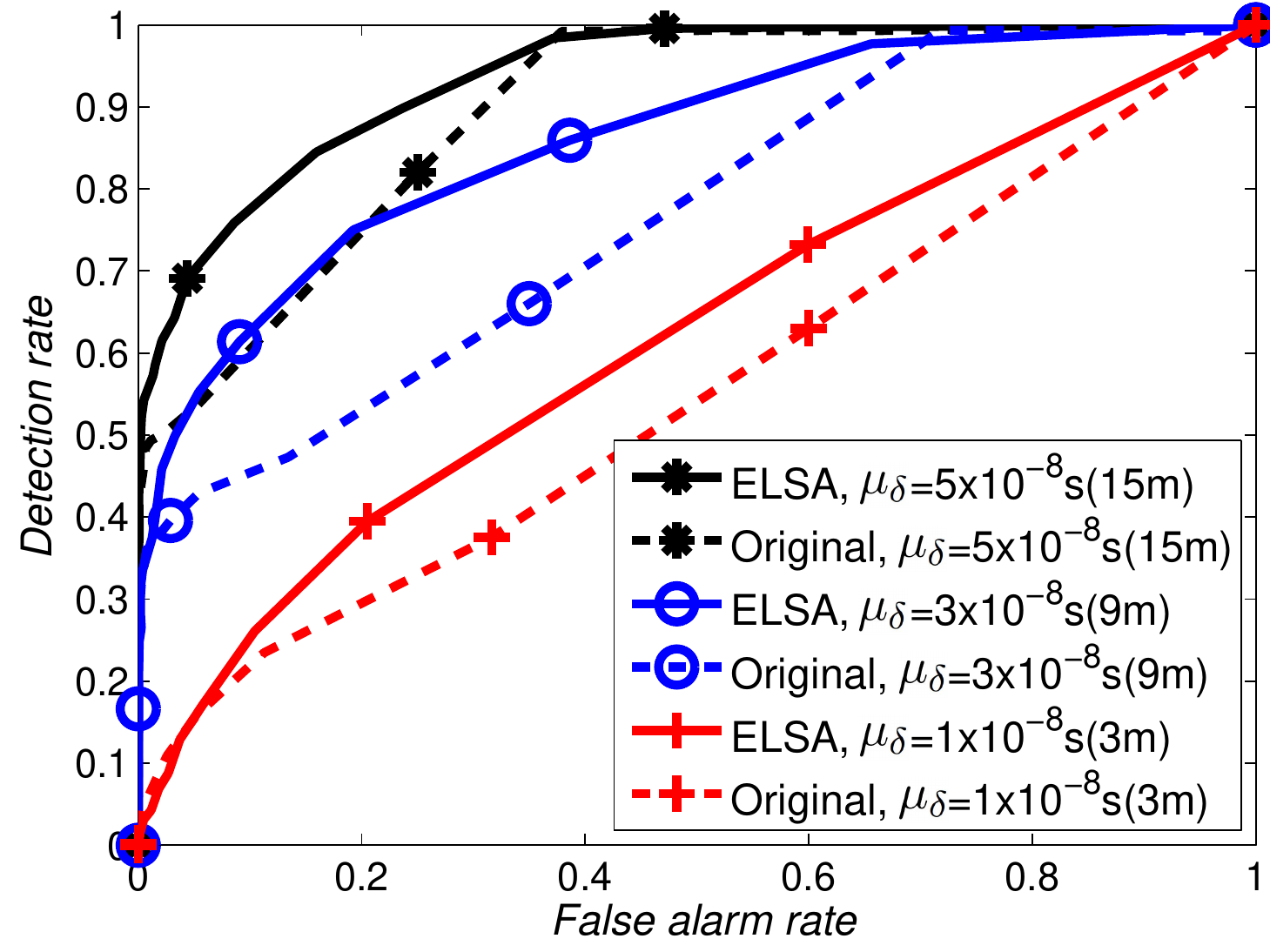}
\caption{ROC curves for different attack mean $\mu_{\delta}$ with three anchors.}
\label{fig:toa_varymean_roc}
\end{figure}

\begin{figure}[t]
\centering
\includegraphics[width=3in]{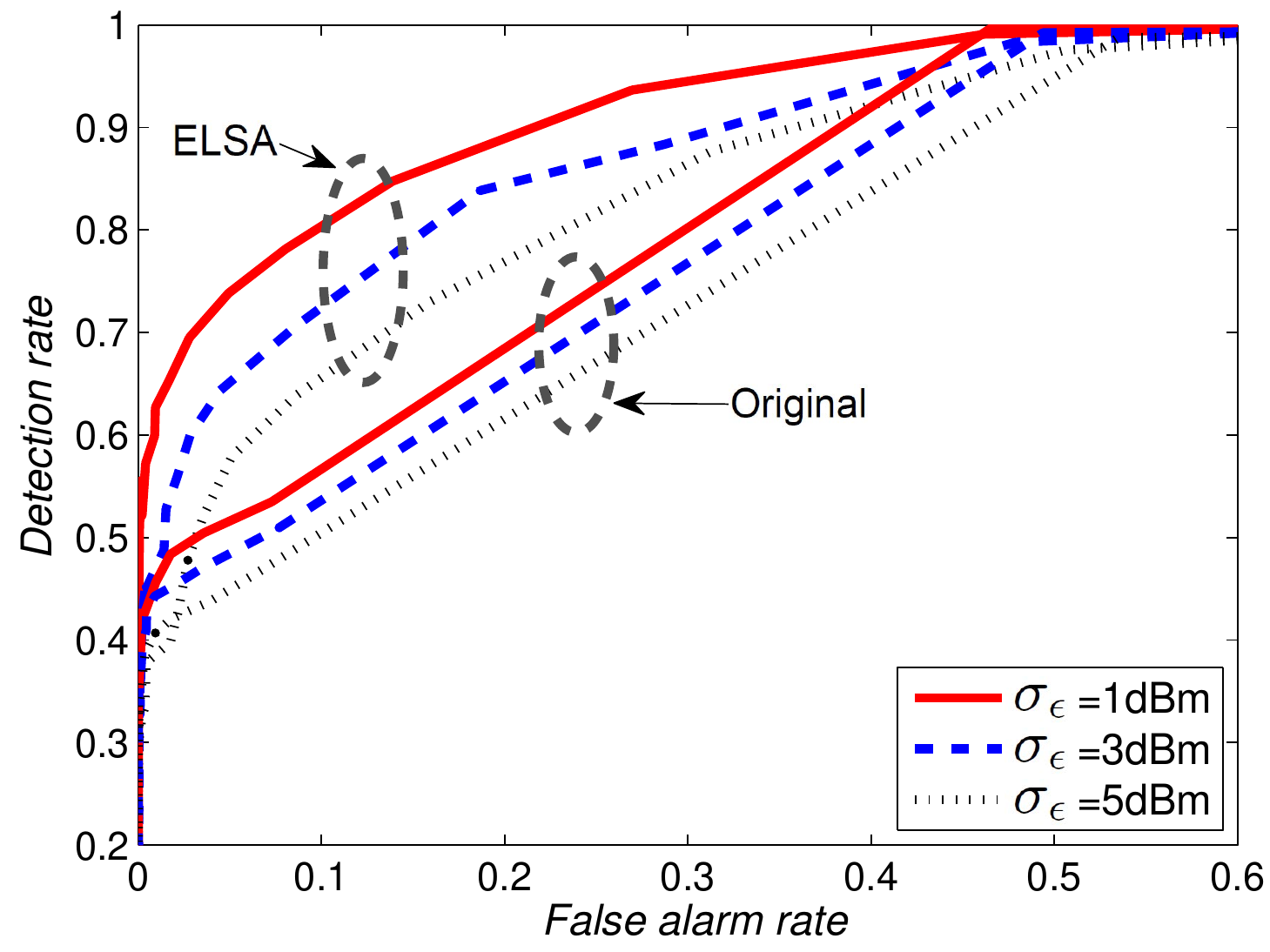}
\caption{ROC curves for different RSS noise variance $\sigma_\epsilon^2$ with  three anchors.}
\label{fig:toa_varyrss_roc}
\end{figure}

\begin{figure}[t]
\centering
\includegraphics[width=3in]{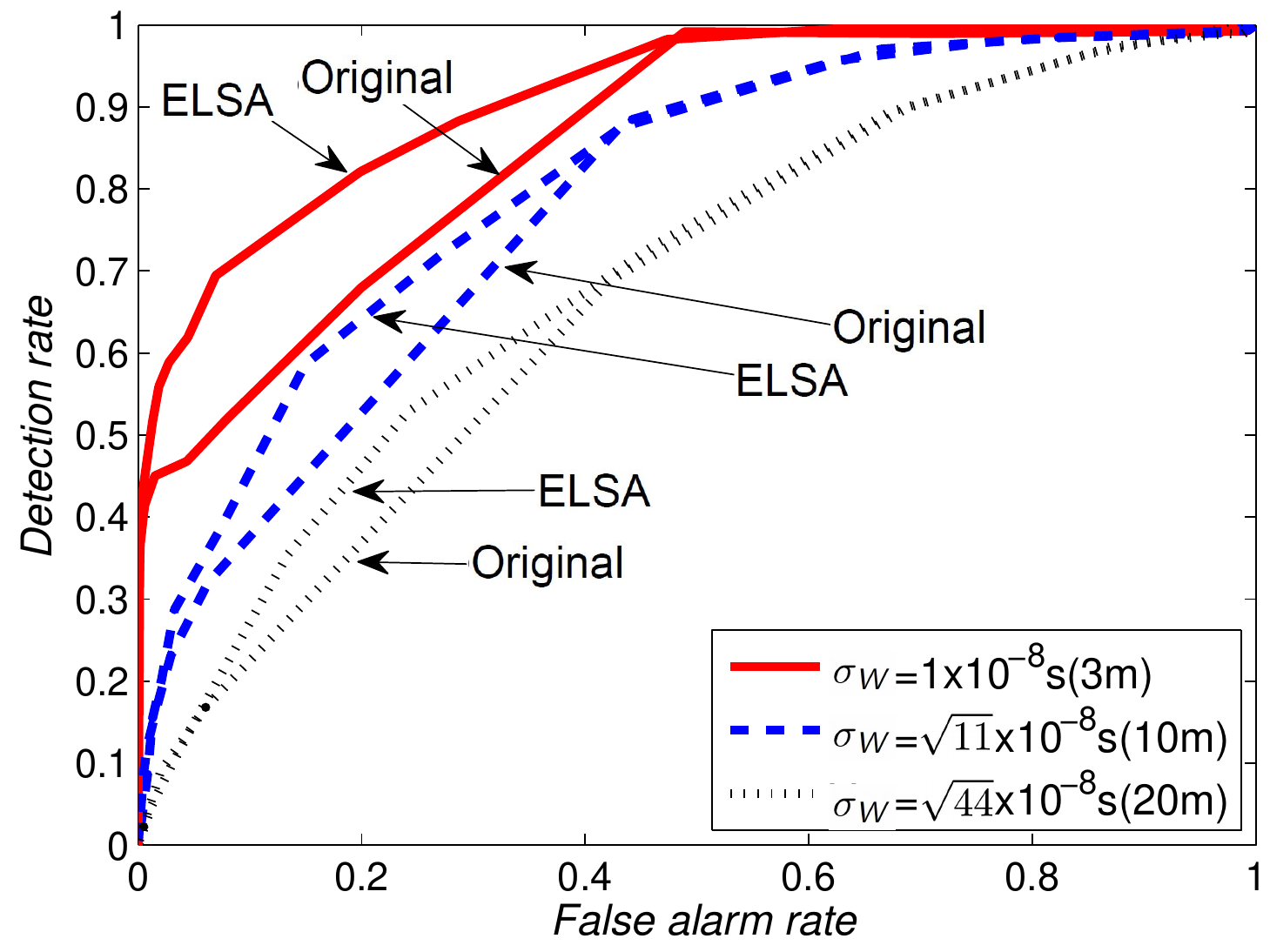}
\caption{ROC curves for different TOA noise variance $\sigma_W^2$ with three anchors.}
\label{fig:toa_varytoa_roc_varyloc}
\end{figure}

\begin{figure}[t]
\centering
\includegraphics[width=3in]{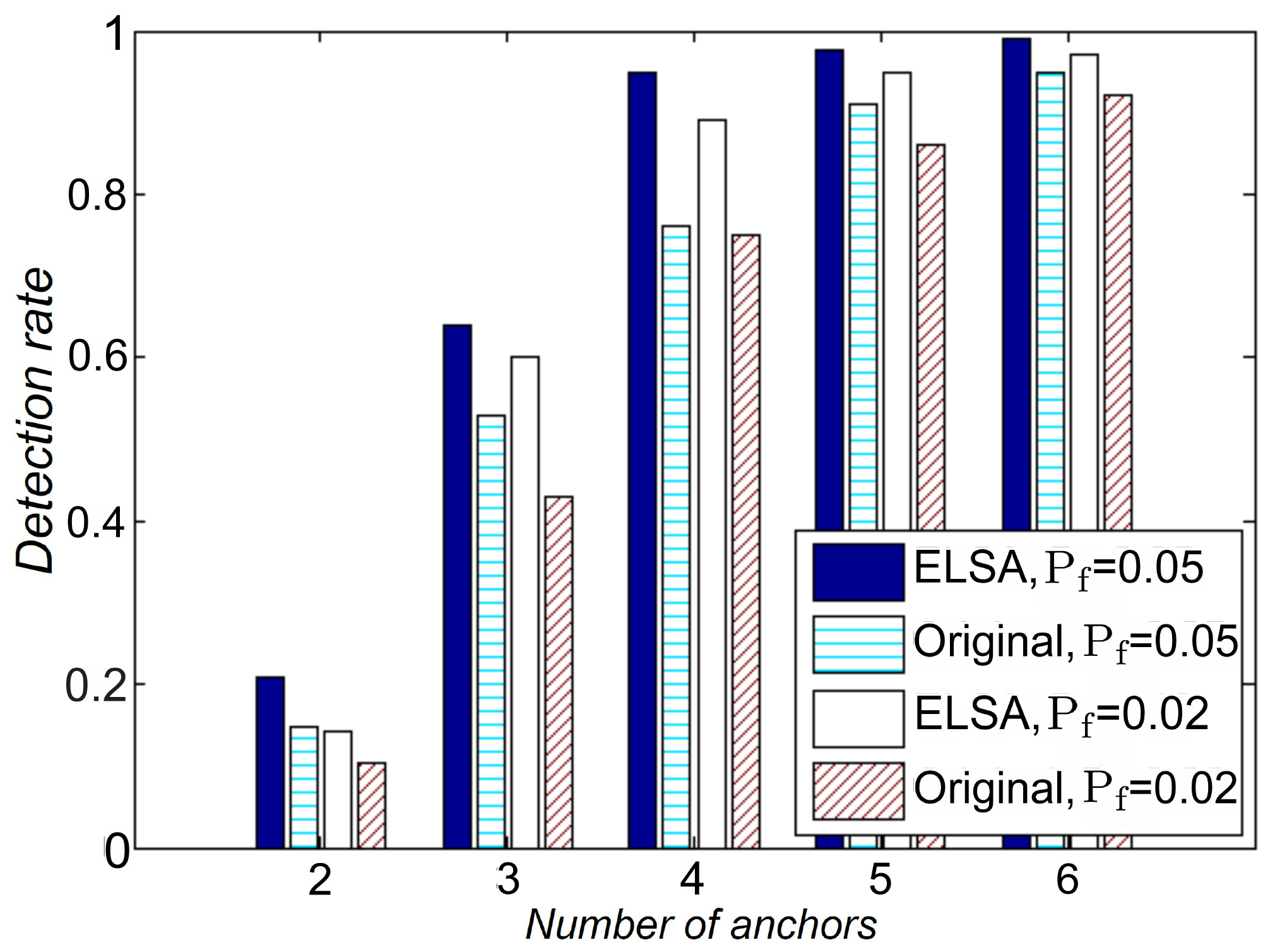}
\caption{Detection rates for different number of anchors (synthetic data) with fixed false alarm rates (${\rm P_ f}$) of 0.02 and 0.05.}
\label{fig:toa_roc_diff_anchors}
\end{figure}

\subsubsection{\textit{ROC Curve Performance under Different Conditions}}
Next, we evaluate the performance of the GLRT tests for different $\mu_{\delta}, \sigma_\epsilon, \sigma_W$ parameters and randomize the target locations for each iteration.
The chosen signal receiving threshold  $\lambda$ includes different inaudible scenarios depending on the target location.
In Fig.~\ref{fig:toa_varymean_roc}, we plot the ROC curves for different attacker delay mean $\mu_{\delta}$ values.
A higher $\mu_{\delta}$ value will perturb the delay measurements further and increase the spoofed distance of the target at the expense of increased detection rate by the GLRT.
Similar to Fig.~\ref{fig:toa_roc}, the detection performance of the conventional GLRT is worse than ELSA's. As \mbox{$\mu_{\delta} >$ \SI{5e-8}{s}} (\SI{15}{m} approx. - take the delay and multiply it with $v_p$), the detection rate for ELSA goes nearer to 100\% and thus we do not plot further.

The impact of obstacles and multipaths can affect the detection performances of the proposed test by increasing the TOA observation noise variances~\cite{patwari2003}. Similarly, the RSS variances will also increase due to the shadowing and multipath.
In Figs.~{\ref{fig:toa_varyrss_roc}} and {\ref{fig:toa_varytoa_roc_varyloc}}, we vary the RSS noise variance $\sigma^2_{\epsilon}$ and TOA noise variance $\sigma^2_{W}$ respectively to verify that the proposed ELSA can still function correctly under large noise variances.
In Fig.~\ref{fig:toa_varyrss_roc}, we vary the RSS noise variance $\sigma^2_{\epsilon}$ and verify that the proposed ELSA can still function correctly under large noise variances.
Note that the performance of the conventional non-audibility aware GLRT is largely unaffected by the RSS noise variance.
However, the performance of ELSA depends more heavily on the RSS readings which affects the audibility information.
In Fig.~\ref{fig:toa_varytoa_roc_varyloc}, we vary the TOA noise variance $\sigma^2_{W}$.
The detection rates for both tests drops as $\sigma^2_{W}$ increases because the attacker's delay is covered by in the TOA observation noise.
Hence, the impact of the attack also drops when the $\sigma^2_{W}$ is high.
Next, we increase the number of deployed anchors and plot the detection performance in
Fig.~\ref{fig:toa_roc_diff_anchors} for fixed false alarm rates. 
We placed an anchor at each corner of the \SI{100}{m} $\times$ \SI{100}{m} area and another two anchors in the middle.
Similarly, the detection rate of the conventional approach is less than the proposed ELSA's as it does not account for audibility.
However, the detection performances for both tests will improve with diminishing returns as the number of anchors increases.

In the event where a malicious node colludes with another node to create a fake audibility condition, the malicious node may either appear to be closer or further to some anchors.
However, our existing threat model which accounts for an i.i.d. adversarial delay (see (5) of Section II-D (Threat Model)) will be able to detect the location spoofing attempt due to the inconsistency in the TOA measurements and RSS readings.
For jamming scenarios,  an adversary may be able to fool the detection test into having a false alarm but he is still unable to successfully spoof his location which is the main goal of the location spoofing detection test.
However, with the emergence of the Ultra Wide Band (UWB) technology, the threat of jamming attacks have been reduced. UWB IoT chip markers (e.g.,~\cite{decawaveweb}) have even claimed that their devices are immune to multipath interference.

\subsection{Results from Real-world Dataset}
We adopt a real sensor network TOA and RSS measurements dataset used in Patwari \textit{et al.}'s works \cite{patwari2003, patwari2007} to validate our proposed audibility framework. The considered network consisted of 44  sensor nodes distributed in an office area in Motorola Labs' Florida Communications Research Lab, in Plantation, FL. Both TOA and RSS measurements were recorded between each sensor node and a high SNR was maintained throughout the experiment to ensure the reliability of the recorded data.
Additional implementation details can be found in the paper \cite{patwari2003} and the dataset is available from the author's website \cite{datasetweb}.
We set the minimum signal receiving threshold $\lambda$ to add inaudible scenarios and evaluated the performances of ELSA and the conventional approach under different scenarios.
We use three of the anchors (node numbers 10, 35, 44) as used by the original authors and an attacker mean of \mbox{$\mu_{\delta} =$ \SI{1.5e-8}{s}} (\SI{4.5}{m} approx.). The anchors are located at the corners of the testbed.

\begin{figure}[t]
\centering
\includegraphics[width=3in]{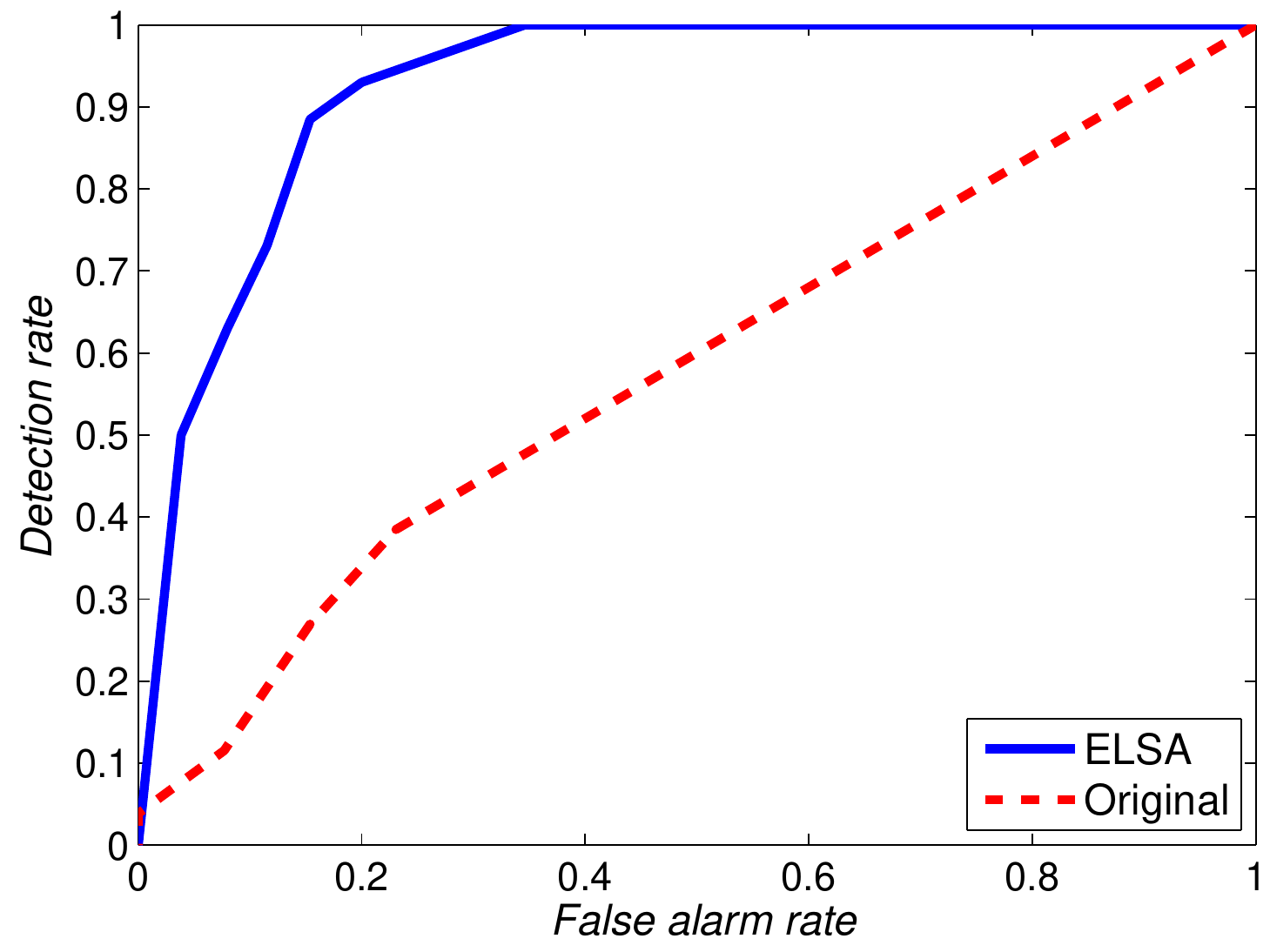}
\caption{ROC curves with $\lambda =$ \SI{-61}{dBm} and 41 different target locations (real-world dataset) and three anchors. }
\label{fig:toa_verify_roc}
\end{figure}

\begin{figure}[t]
\centering
\includegraphics[width=3in]{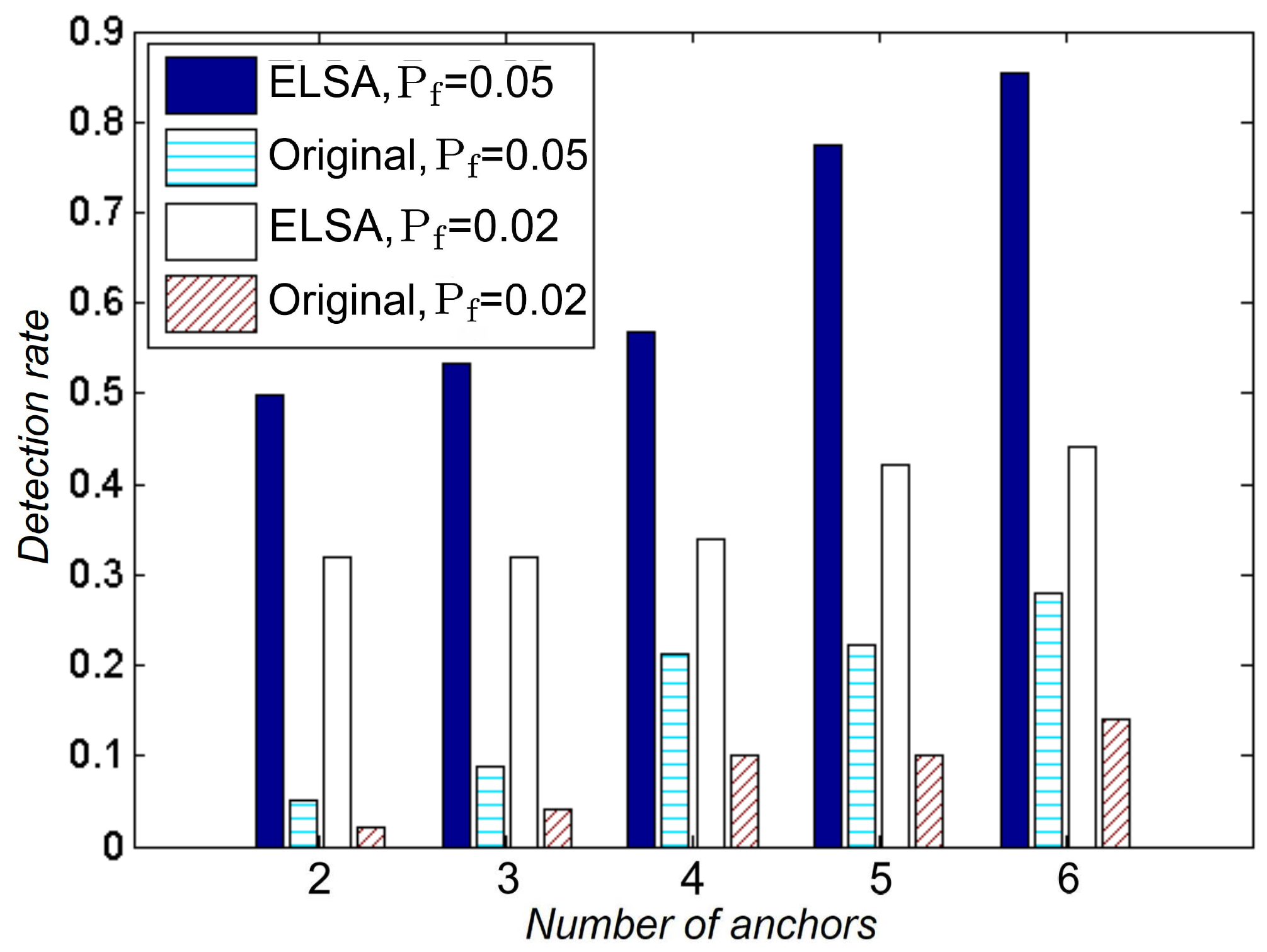}
\caption{Detection rates for different number of anchors (real-world dataset) with $\lambda =$ \SI{-61}{dBm} for false alarm rates (${\rm P_ f}$) of 0.02 and 0.05.}
\label{fig:toa_verify_diff_anchors}
\end{figure}

\textit{ROC Curve Performance:}
In Fig.~\ref{fig:toa_verify_roc}, we plot the ROC curves for $\lambda = -$\SI{61}{dBm}. The chosen scenario includes a good mix of different numbers of audible anchors and highlights the superiority of ELSA compared to the conventional GLRT.
For a fixed false alarm rate, ELSA has a significantly higher detection rate. 
The ROC curve for the conventional GLRT however, is closer to the diagonal line (not drawn) at low false alarm rates which indicates its poorer detection rate trade-off. 
A higher $\mu_{\delta}$ parameter will lead to a steeper ROC curve for both schemes with the proposed scheme still being superior.
In Fig.~\ref{fig:toa_verify_diff_anchors}, we vary the number of deployed anchors and plot the detection rates of the tests for a fixed false alarm rate.
Note that the relative detection improvements of ELSA is significantly better that the relative detection improvements taken from the synthetic results in Fig.~\ref{fig:toa_roc_diff_anchors}. 
This could be due to the limited target locations and their clustered distribution in the dataset whereas in our simulation, we uniformly picked the location of each target in each iteration.

\section{Conclusion}
A new audibility-based framework has been introduced in this paper for detecting location spoofing attempts.
We showed how the conventional TOA-based method may not be able to detect location spoofing attempts especially during inaudible scenarios and developed an audibility-aware detection test called ELSA to do so.
ELSA is able to overcome outage scenarios by exploiting their implicit audibility information.
In addition, we have also demonstrated that ELSA has a better detection performance compared to the conventional GLRT using experimental results from both simulations and a real-world data set.
ELSA also accommodates usage of low-cost IoT devices and lessens the need to deploy a dense network of anchors.
A future research direction would be to investigate other deployment environment-specific TOA, RSS-based or even energy harvesting models to further improve existing detection performances.

\appendix

\label{Appendix_woaudibility1}
\subsection{GLRT Test Statistic without Audibility Considerations}
Consider the case where only $l$ out of the $n$ deployed anchors receive a delay measurement from the target.
Under the null hypothesis $\mathcal{H}_0$, the likelihood function is simply
	\begin{equation}
	p( {\bf t}| \mathcal{H}_0, \widehat{\mathbf{\Theta}}_{\MAP}^{\mathcal{H}_0} )	= \prod^{l}_{i = 1} \mathcal{N}(t_i;  \frac{d(\widehat{\mathbf{\Theta}}_{\MAP}^{\mathcal{H}_0}, {\bf x}_i)}{v_p}, \sigma^2_{W}).
	\end{equation}
Under the alternative hypothesis $\mathcal{H}_1$, the likelihood function is given by
	\begin{equation}
	p( {\bf t}| \mathcal{H}_1, \widehat{\mathbf{\Theta}}_{\MAP}^{\mathcal{H}_1} ) = \prod^{n}_{i = 1}
\mathcal{N}(t_i;  \frac{d(\widehat{\mathbf{\Theta}}_{\MAP}^{\mathcal{H}_1}, {\bf x}_i)}{v_p} + \mu_{\delta}, \sigma^2_{W} + \sigma^2_{\delta}).
	\end{equation}
We obtain the test statistic
	\begin{align*}
	&\Lambda({\bf t})
	= 	\frac{ p( {\bf t}| \mathcal{H}_1, \widehat{\mathbf{\Theta}}_{\MAP}^{\mathcal{H}_1})                        }{
p( {\bf t}| \mathcal{H}_0, \widehat{\mathbf{\Theta}}_{\MAP}^{\mathcal{H}_0})             }
	\\&= \frac{ \prod^{l}_{i=1}  \frac{1}{\sqrt{2\pi} (\sigma_W + \sigma_{\delta})} \exp\{ - \frac{1}{2(\sigma_W^2 + \sigma^2_{\delta})} (t_i - \frac{d(\widehat{\mathbf{\Theta}}_{\MAP}^{\mathcal{H}_1}, {\bf x}_i)}{v_p})^2 \}    }{                           \prod^{l}_{i=1}  \frac{1}{\sqrt{2\pi} \sigma_W} \exp\{ - \frac{1}{2\sigma_W^2} (t_i - \frac{d(\widehat{\mathbf{\Theta}}_{\MAP}^{\mathcal{H}_0}, {\bf x}_i)}{v_p})^2 \}
}
	 \quad
	\underset{\mathcal{H}_0}{\overset{\mathcal{H}_1}{\gtrless}}
 	\eta. \tageq
	\end{align*}

\subsection{Derivation of MAP Estimate}
\label{Appendix_mapestimate}
The MAP estimate is given by
	% \begin{equation*}
	% \begin{split}
	% \widehat{\mathbf{\Theta}}_{\MAP}^{\mathcal{H}_j}
	% &= \argmax_{\mathbf{\Theta}}  p(\mathbf{\Theta} | {\bf t}, {\bf r}, \mathcal{H}_j)
	% \\
	% &= \argmax_{\mathbf{\Theta}}  p({\bf t}, {\bf r} | \mathbf{\Theta}, \mathcal{H}_j) p(\mathbf{\Theta} | \mathcal{H}_j)
	% \\
	% &= \argmax_{\mathbf{\Theta}}  p({\bf t} | {\bf r}, \mathbf{\Theta}, \mathcal{H}_j) P({\bf r} | \mathbf{\Theta}, \mathcal{H}_j) p(\mathbf{\Theta} | \mathcal{H}_j)
	% \end{split}
	% \end{equation*}

	\begin{align*}
	\widehat{\mathbf{\Theta}}_{\MAP}^{\mathcal{H}_j}
	&= \argmax_{\mathbf{\Theta}}  p(\mathbf{\Theta} | {\bf t}, {\bf r}, \mathcal{H}_j)
	\\
	&= \argmax_{\mathbf{\Theta}}  p({\bf t}, {\bf r} | \mathbf{\Theta}, \mathcal{H}_j) p(\mathbf{\Theta} | \mathcal{H}_j)
	\\
	&= \argmax_{\mathbf{\Theta}}  p({\bf t} | {\bf r}, \mathbf{\Theta}, \mathcal{H}_j) P({\bf r} | \mathbf{\Theta}, \mathcal{H}_j) p(\mathbf{\Theta} | \mathcal{H}_j)
	\\
	&= \argmax_{\mathbf{\Theta}}  \prod_{i=1}^{n}
\Big[  p( t_i | r_i, \mathbf{\Theta}, \mathcal{H}_j)\mathds{1}(r_i = 1) + \mathds{1}(r_i = 0)
\Big]
	\\
	& \quad \times
p(r_i | \mathbf{\Theta}, \mathcal{H}_j) p(\mathbf{\Theta} | \mathcal{H}_j)
	\\
	%
	% NEW
 &= \argmax_{\mathbf{\Theta}}  \sum_{i=1}^{n}
\log
	\Big[
p( t_i | r_i, \mathbf{\Theta}, \mathcal{H}_j) \mathds{1}(r_i = 1)
\\
	& \quad
+  \mathds{1}(r_i = 0)
\Big]
+
\log p(r_i | \mathbf{\Theta}, \mathcal{H}_j)
p(\mathbf{\Theta} | \mathcal{H}_j)
	\\
	&= \argmax_{\mathbf{\Theta}}  \sum_{i=1}^{n}
	\Big[ \log p( t_i | r_i, \mathbf{\Theta}, \mathcal{H}_j) \mathds{1}(r_i = 1)
	\\*
	& \quad +
 \log p(r_i | \mathbf{\Theta}, \mathcal{H}_j)
\Big] + \log p(\mathbf{\Theta} | \mathcal{H}_j)
	\\
	&= \argmax_{\mathbf{\Theta}}  \sum_{i=1}^{n}
	 \log \mathcal{N}(t_i;  \frac{d(\mathbf{\Theta}, {\bf x}_i)}{v_p} +\delta_i, \sigma^2_{W})
 \mathds{1}(r_i = 1)
	\\*
	& \,\,\,\,\,\, + \sum_{i=1}^{n}
\log P(r_i = 1 | \mathbf{\Theta}, \mathcal{H}_j) \mathds{1}(r_i = 1)
	\\*
	& \quad
+ P(r_i = 0 | \mathbf{\Theta}, \mathcal{H}_j) \mathds{1}(r_i = 0)
+
\log p(\mathbf{\Theta} | \mathcal{H}_j), \tageq
	\label{eqn:mapestimate}
	\end{align*}

\noindent
where 
$
\log
	\Big[
p( t_i | r_i, \mathbf{\Theta}, \mathcal{H}_j) \mathds{1}(r_i = 1)
+  \mathds{1}(r_i = 0)
\Big] $

\noindent

 \quad\quad\quad $=
\begin{cases}
    \log p( t_i | r_i, \mathbf{\Theta}, \mathcal{H}_j) & \text{if } r_i = 1, \\
    \log(1) = 0              & \text{if }r_i = 0.
\end{cases}
	$

\noindent
The indicator function $\mathds{1}(\cdot)$ is used to ensure that the product is non-zero when no delay measurements are received.
The probability of
receiving the observed delay measurement by anchor $i$ (given that the node is audible) is given by:
	\begin{equation*}
	p( t_i | r_i = 1, \mathbf{\Theta})  = \mathcal{N}(t_i;  \frac{d(\mathbf{\Theta}, {\bf x}_i)}{v_p} +\delta_i, \sigma^2_{W}),
	\end{equation*}
and the probability of anchor $i$ receiving a signal with a RSS value that is greater or equal to the minimum signal receiving threshold $\lambda$ is given by:
	\begin{equation*}
	\begin{split}
	P(r_i = 1 | \mathbf{\Theta})  & = \int^{\infty}_{\lambda}
	\mathcal{N}(j;   P_{t} - 10 \alpha \log \frac{d(\mathbf{\Theta}, {\bf x}_i)}{d_0}, \sigma^2_{\epsilon})   \, dj
\\
 	& = 1 - \Phi \left( \frac{\lambda - P_{t} + 10 \alpha \log \frac{d(\mathbf{\Theta}, {\bf x}_i)}{d_0} } { \sigma_{\epsilon} } \right).
	\end{split}
	\end{equation*}
On the other hand, we can compute the probability that anchor $i$ does not receive a delay measurement, which is given by
	\begin{equation*}
	P(r_i = 0 | \mathbf{\Theta}) =
 1 - P(r_i = 1 | \mathbf{\Theta}).
	\end{equation*}
Therefore, the generalized likelihood function can be expressed as:
	\begin{align*}
	&p( {\bf t}, {\bf r} |\mathcal{H}_j, \widehat{\mathbf{\Theta}}_{\MAP}^{\mathcal{H}_j} ) 		
=
p( {\bf t}| {\bf r}, \mathcal{H}_j, \widehat{\mathbf{\Theta}}_{\MAP}^{\mathcal{H}_j})
p( {\bf r}| \mathcal{H}_j, \widehat{\mathbf{\Theta}}_{\MAP}^{\mathcal{H}_j})
	\\
&=
	\prod^{n}_{i = 1}  \left[
	\mathcal{N}(t_i;   \frac{d(\widehat{\mathbf{\Theta}}_{\MAP}^{\mathcal{H}_j}, {\bf x}_i)}{v_p}, \sigma^2_{W})  \mathds{1}(r_i = 1) + \mathds{1}(r_i = 0) \right]
	\\
&    \quad
	\times
	\prod^{n}_{i = 1}  \Bigg[ P(r_i = 1 | \widehat{\mathbf{\Theta}}_{\MAP}^{\mathcal{H}_j}) \mathds{1}(r_i = 1)
	+ 	P(r_i = 0 | \widehat{\mathbf{\Theta}}_{\MAP}^{\mathcal{H}_j}) \mathds{1}(r_i = 0) \Bigg]
	\\
&=
	\prod^{n}_{i = 1}  \left[
	\mathcal{N}(t_i;   \frac{d(\widehat{\mathbf{\Theta}}_{\MAP}^{\mathcal{H}_j}, {\bf x}_i)}{v_p}, \sigma^2_{W})  \mathds{1}(r_i = 1) + \mathds{1}(r_i = 0) \right]
	\\
&   \quad
	\times
	\prod^{n}_{i = 1}
	\left[
	1 - \Phi \left( \frac{\lambda - P_{t} + 10 \alpha \log \frac{d(\widehat{\mathbf{\Theta}}_{\MAP}^{\mathcal{H}_j}, {\bf x}_i)}{d_0} } { \sigma_{\epsilon} } \right) \right] \mathds{1}(r_i = 1)
	\\
	& \quad
+
	\Phi \left( \frac{\lambda - P_{t} + 10 \alpha \log \frac{d(\widehat{\mathbf{\Theta}}_{\MAP}^{\mathcal{H}_j}, {\bf x}_i)}{d_0} } { \sigma_{\epsilon} } \right)
\mathds{1}(r_i = 0)  .
	\end{align*}

\subsection{Proof of Theorem 1}
\label{Appendix_theorem}
We let the distance related terms
	$\psi_i = \frac{d(\mathbf{\Theta}, {\bf x}_i)}{v_p}$,
	$\psi_i' = \frac{  d(\mathbf{\Theta}, {\bf x}_i)  }{d_0}$,  and
 	$\mu' = \frac{\mu_{\delta}}{d_0}$.
It can be shown that the detection and false alarm rates for the conventional GLRT without audibility considerations are given by
\[
	P_{d | \psi}^{\rm NA} =
	\int_\gamma^{\infty}    p(z  |  \mathcal{H}_1, \psi   ) dz
	\;\;\text{and}\;\;
	P_{f | \psi}^{\rm NA} =
	\int_\gamma^{\infty}    p(z  |  \mathcal{H}_0, \psi   )   dz,
\]
	% \begin{equation*}
	% P_{d | \psi}^{\rm NA} =
	% \int_\gamma^{\infty}    p(z  |  \mathcal{H}_1, \psi   ) dz,
	% \end{equation*}
	% \begin{equation*}
	% P_{f | \psi}^{\rm NA} =
	% \int_\gamma^{\infty}    p(z  |  \mathcal{H}_0, \psi   )   dz,
	% \end{equation*}
respectively (see Appendix-\ref{Appendix_woaudibility2}), where $\gamma$ is a threshold, and $l$ is the number of received delay measurements.
On the other hand, the detection and false alarm rates for the proposed audibility-aware GLRT are given by
\[
	P_{d | \psi}^{\rm A} =
	\int_\gamma^{\infty} \!\! p(z + \Xi |  \mathcal{H}_1, \psi   ) dz
	\;\;\text{and}\;\;
	P_{f | \psi}^{\rm A} =
	\int_\gamma^{\infty}   \!\! p(z  |  \mathcal{H}_0, \psi   )   dz,
\]
	% \begin{equation*}
	% P_{d | \psi}^{\rm A} =
	% \int_\gamma^{\infty}    p(z + \Xi |  \mathcal{H}_1, \psi   ) dz,
	% \end{equation*}
	% \begin{equation*}
	% P_{f | \psi}^{\rm A} =
	% \int_\gamma^{\infty}    p(z  |  \mathcal{H}_0, \psi   )   dz,
	% \end{equation*}
respectively (see Appendix-\ref{Appendix_audibility}) where the term $\Xi$ (from \eqref{eqn:xi}) consists of the audibility-related probabilities.
Note that the false alarm rates for both cases are the same:
	\begin{equation*}
	P_{f | \psi}^{\rm A} = P_{f | \psi}^{\rm NA} =
	\int_\gamma^{\infty}    p(z  |  \mathcal{H}_0, \psi   )   dz.
	\end{equation*}
Hence, it can be seen that for a fixed false alarm rate $P_{f | \psi} $, the detection rates
	\begin{equation*}
	P_{d | \psi}^{\rm A} \geq P_{d | \psi}^{\rm NA},
	\end{equation*}
if $\Xi \leq 0$ holds since the complementary cdf function in both
$	P_{d | \psi}^{\rm A} $ and $ P_{d | \psi}^{\rm NA}$ is a non-increasing function.

Suppose that $\Xi \leq 0$ and $\mu_{\delta} > 0$ (which is true in our model).
From \eqref{eqn:xi}, the $\Xi$ term can be expressed as \eqref{eqn:xiterms}.

Next, we simplify the equation to obtain:
	\begin{equation}
	\begin{aligned}[b]
	\Xi &= \sum^{l}_{i=1} \ln    \frac{                1 - \Phi \left( \frac{\lambda - P_{t} + 10 \alpha \log (\psi'_i+\mu') } { \sigma_{\epsilon} } \right)
	}{1 - \Phi \left( \frac{\lambda - P_{t} + 10 \alpha \log \psi'_i } { \sigma_{\epsilon} } \right)}
\\&
	+ \sum^{n}_{i=l+1} \ln    \frac { \Phi \left( \frac{\lambda - P_{t} + 10 \alpha \log (\psi'_i-\mu') } { \sigma_{\epsilon} } \right)
	}{\Phi \left( \frac{\lambda - P_{t} + 10 \alpha \log \psi'_i } { \sigma_{\epsilon} } \right)   }
	\leq 0.
	\end{aligned}
	\end{equation}
Since  $\mu_{\delta} > 0$, then $\mu' = \frac{\mu_{\delta}}{d_0} > 0$ as $d_0 > 0$.
Because the logarithm function is strictly increasing for positive inputs, we have
	\begin{equation*}
 \Phi \left( \log ((\psi'_i-\mu')^{+})  \right)
 <
 \Phi \left( \log (\psi'_i)  \right)
 <
 \Phi \left( \log (\psi'_i+\mu')  \right).
	\end{equation*}
Note that $(\psi'_i-\mu')^{+}$ is strictly positive as it is not possible to receive a negative delay.
Similarly, this implies that
	\begin{equation*}
	\frac{                1 - \Phi \left( \frac{\lambda - P_{t} + 10 \alpha \log (\psi'_i+\mu') } { \sigma_{\epsilon} } \right)
	}{1 - \Phi \left( \frac{\lambda - P_{t} + 10 \alpha \log \psi'_i } { \sigma_{\epsilon} } \right)} < 1 ,
	\end{equation*}
and
	\begin{equation*}
	\frac { \Phi \left( \frac{\lambda - P_{t} + 10 \alpha \log (\psi'_i-\mu') } { \sigma_{\epsilon} } \right)
	}{\Phi \left( \frac{\lambda - P_{t} + 10 \alpha \log \psi'_i } { \sigma_{\epsilon} } \right)   }  < 1.
	\end{equation*}
Since the natural logarithmic function always has a negative value when the inputs are less than 1 and $\Xi$ consists of the summation of negative terms, hence, the statement $\Xi \leq 0$ must be true and
$P_{d | \psi}^{\rm A} \geq P_{d | \psi}^{\rm NA}$.

Subsequently, we can marginalize $P_{d | \psi_i}$ over all possible $\psi_i$ values to obtain
	\begin{equation*}
	P_d = \int P_{d | \psi_i}  \times p(\psi_i) \, d\psi_i .
	\end{equation*}
Therefore, for a fixed $P_f$, the following inequalities hold:
	\begin{equation*}
	P_d^{\rm A}  \geq P_d^{\rm NA},
	\end{equation*}
since their equivalent representations,
	\begin{equation*}
	\begin{split}
	\int_{\psi}
	\int_\gamma^{\infty}    p(z + \Xi |  \mathcal{H}_1, \psi   )       p(\psi_i) \,\,  dz d\psi_i
\\
\geq
\int_{\psi}
	\int_\gamma^{\infty}    p(z  |  \mathcal{H}_1, \psi   )      p(\psi_i)   dz d\psi_i   ,
	\end{split}
	\end{equation*}
where the following has already been proven to be true:
	\begin{equation*}
	\int_\gamma^{\infty}    p(z + \Xi |  \mathcal{H}_1, \psi   )   dz
\geq
	\int_\gamma^{\infty}    p(z  |  \mathcal{H}_1, \psi   )   dz .
	\end{equation*}
Hence, we complete the proof.
\hfill\IEEEQEDhere

\begin{figure*}
\begin{minipage}{\textwidth}
	\begin{equation}
	\begin{aligned}[b]
	\Xi = &2\sigma_W^2 (\sigma_W^2 +    \sigma^2_{\delta})  \Bigg[     \sum^{l}_{i=1} \ln                    \bigg(1 - \Phi \left( \frac{\lambda - P_{t} + 10 \alpha \log (\psi'_i+\mu') } { \sigma_{\epsilon} } \right) \bigg)
	+ \sum^{n}_{i=l+1} \ln    \Phi \left( \frac{\lambda - P_{t} + 10 \alpha \log (\psi'_i-\mu') } { \sigma_{\epsilon} } \right)
	\\&
	-  \sum^{l}_{i=1} \ln                            \bigg(1 - \Phi \left( \frac{\lambda - P_{t} + 10 \alpha \log \psi'_i } { \sigma_{\epsilon} } \right) \bigg)
	- \sum^{n}_{i=l+1} \ln   \Phi \left( \frac{\lambda - P_{t} + 10 \alpha \log \psi'_i } { \sigma_{\epsilon} } \right)     	\Bigg]  .
	\end{aligned}
\label{eqn:xiterms}
	\end{equation}
\end{minipage}

\hspace{\fill}\line(1,0){500}\hspace{\fill}
\end{figure*}

\subsection{Derivation of Detection and False Alarm Probabilities without  Audibility Considerations}
\label{Appendix_woaudibility2}

We denote the distance-related term as
	\begin{equation}
	\psi_i = \frac{d(\mathbf{\Theta}, {\bf x}_i)}{v_p}.
	\label{eqn:assumptionk}
	\end{equation}
We obtain the test statistic which does not take into account audibility as follows:
	\begin{equation}
	\begin{split}
	&\Lambda({\bf t})
	\\
	&= \frac{ \prod^{l}_{i=1}  \frac{1}{\sqrt{2\pi} (\sigma_W + \sigma_{\delta})} \exp\{ - \frac{1}{2 (\sigma_W^2 + \sigma^2_{\delta})} (t_i - \psi_i - \mu_{\delta})^2 \}    }{                           \prod^{l}_{i=1}  \frac{1}{\sqrt{2\pi} \sigma_W} \exp\{ - \frac{1}{2\sigma_W^2} (t_i - \psi_i)^2 \}
}
	\underset{\mathcal{H}_0}{\overset{\mathcal{H}_1}{\gtrless}}
 	\eta.
	\end{split}
	\end{equation}
Taking the logarithm on both sides, we obtain \eqref{eqn:deriveteststataudibility}.
\begin{figure*}[!t]
\begin{minipage}{\textwidth}
	\begin{equation}
	\begin{aligned}[b]
	\Lambda({\bf t}) &= \sum^{l}_{i=1}  \ln  \frac{1}{\sqrt{2\pi} (\sigma_W + \sigma_{\delta})}
- \sum^{l}_{i=1}   \frac{(t_i - \psi_i - \mu_{\delta})^2 }{2 (\sigma_W^2 +    \sigma^2_{\delta})}
- \sum^{l}_{i=1}   \ln  \frac{1}{\sqrt{2\pi} \sigma_W}
+ \sum^{l}_{i=1}   \frac{(t_i - \psi_i)^2 }{2\sigma_W^2}
	\\&=
\ln  \frac{ \sigma_W }{(\sigma_W + \sigma_{\delta})}
- \sum^{l}_{i=1}   \frac{(t_i - \psi_i - \mu_{\delta})^2 }{2 (\sigma_W^2 +    \sigma^2_{\delta})}
+ \sum^{l}_{i=1}   \frac{(t_i - \psi_i)^2 }{2\sigma_W^2}
	\\&=
\ln  \frac{\sigma_W }{(\sigma_W + \sigma_{\delta})}
+ \sum^{l}_{i=1}   \frac{ (\sigma_W^2 +    \sigma^2_{\delta})
(t_i^2 + \psi_i^2 - 2t_i \psi_i)
- \sigma_W^2 (t_i^2+ \mu_{\delta}^2  + \psi_i^2 - 2\psi_i t_i - 2t_i \mu_{\delta} + 2\psi_i \mu_{\delta}) }{2\sigma_W^2 (\sigma_W^2 +    \sigma^2_{\delta})}
	\\&=
\ln  \frac{\sigma_W }{(\sigma_W + \sigma_{\delta})}
+ \sum^{l}_{i=1}
\frac{  \sigma^2_{\delta} t_i^2 + 2 \mu_{\delta} \sigma_W^2 t_i - 2\sigma^2_{\delta} \psi_i t_i + \sigma^2_{\delta} \psi_i^2  - 2 \psi_i \mu_{\delta} \sigma_W^2 - \mu_{\delta}^2 \sigma_W^2  }{2\sigma_W^2 (\sigma_W^2 +    \sigma^2_{\delta})}
	\underset{\mathcal{H}_0}{\overset{\mathcal{H}_1}{\gtrless}}
 	 \ln \eta
\\&\hspace{-0.5cm} \text{Next, we shift some terms over to the RHS,}
	\\&
 \sum^{l}_{i=1}
 \sigma^2_{\delta} t_i^2 + 2 \mu_{\delta} \sigma_W^2 t_i - 2\sigma^2_{\delta} \psi_i t_i
	\underset{\mathcal{H}_0}{\overset{\mathcal{H}_1}{\gtrless}}
 	2\sigma_W^2 (\sigma_W^2 +    \sigma^2_{\delta})
\ln  \left(
\frac{\eta(\sigma_W + \sigma_{\delta})  }{ \sigma_W}  \right)
+ \sum^{l}_{i=1}    2 \psi_i \mu_{\delta} \sigma_W^2 + \mu_{\delta}^2 \sigma_W^2
- \sigma^2_{\delta} \psi_i^2
	\\&
 \sum^{l}_{i=1}
 \sigma^2_{\delta} t_i^2 + 2 \mu_{\delta} \sigma_W^2 t_i - 2\sigma^2_{\delta} \psi_i t_i
	\underset{\mathcal{H}_0}{\overset{\mathcal{H}_1}{\gtrless}}
\gamma.
	\end{aligned}
	\label{eqn:deriveteststataudibility}
	\end{equation}
\end{minipage}
\hspace{\fill}\line(1,0){500}\hspace{\fill}
\end{figure*}

Now, let $Z = \sum^{l}_{i=1}
 \sigma^2_{\delta} t_i^2 + 2 \mu_{\delta} \sigma_W^2 t_i - 2\sigma^2_{\delta} \psi_i t_i$
and  $\gamma$ be the threshold.
The detection probability for the non-audibility-aware GLRT is given by
	\begin{equation}
	\begin{split}
	P_{d | \psi}^{\rm NA} &=
	P( z > \gamma  |  \mathcal{H}_1, \psi   )
	=
	\int_\gamma^{\infty}    p(z  |  \mathcal{H}_1, \psi   ) dz,
	\end{split}
	\end{equation}
and the false alarm probability is given by
	\begin{equation}
	\begin{split}
	P_{f | \psi}^{\rm A} &=
	P( z > \gamma  |  \mathcal{H}_0, \psi   )
	 =
	\int_\gamma^{\infty}    p(z  |  \mathcal{H}_0, \psi   )   dz.
	\end{split}
	\end{equation}

\subsection{Derivation of Detection and False Alarm Probabilities with Audibility Considerations}
\label{Appendix_audibility}

From the test statistic derived in \eqref{eqn:glrtaudibiltiy}, we obtain:
	\begin{align*}
	&\Lambda({\bf t}, {\bf r})
	\\&= \Bigg(
 	\prod^{n}_{i = 1}
\mathcal{N}(t_i;  \psi_i + \mu_{\delta}, \sigma^2_{W} + \sigma^2_{\delta})\mathds{1}(r_i = 1) + \mathds{1}(r_i = 0)
	\\& \quad\quad \times
                \prod^{n}_{i = 1}  \Big[ P(r_i = 1 | \widehat{\mathbf{\Theta}}_{\MAP}^{\mathcal{H}_1}) \mathds{1}(r_i = 1)
	\\&
	\quad\quad\quad\quad + 	P(r_i = 0 | \widehat{\mathbf{\Theta}}_{\MAP}^{\mathcal{H}_1}) \mathds{1}(r_i = 0) \Big]         \Bigg)
	\\&  \Bigg/ \Bigg(
\prod^{n}_{i = 1} \mathcal{N}(t_i;  \psi_i, \sigma^2_{W})\mathds{1}(r_i = 1) + \mathds{1}(r_i = 0)
	\\&  \quad\quad \times
	       \prod^{n}_{i = 1}  \Big[ P(r_i = 1 | \widehat{\mathbf{\Theta}}_{\MAP}^{\mathcal{H}_0}) \mathds{1}(r_i = 1)
	\\&
	\quad\quad\quad\quad + 	P(r_i = 0 | \widehat{\mathbf{\Theta}}_{\MAP}^{\mathcal{H}_0}) \mathds{1}(r_i = 0) \Big]   \Bigg)
	\underset{\mathcal{H}_0}{\overset{\mathcal{H}_1}{\gtrless}}
 	\eta. \tageq
	\label{eqn:teststataudibility}
	\end{align*}
We further let  $\psi'_i = \frac{ d(\mathbf{\Theta}, {\bf x}_i)  }{d_0}$ and $\mu' = \frac{\mu_{\delta}}{d_0}$ to obtain the following audibility related equations under $\mathcal{H}_0$:
	\begin{equation}
	\begin{aligned}[b]
	P(r_i = 0 | \widehat{\mathbf{\Theta}}_{\MAP}^{\mathcal{H}_0})
	= \Phi \left( \frac{\lambda - P_{t} + 10 \alpha \log \psi'_i } { \sigma_{\epsilon} } \right),
	\\
	P(r_i = 1 | \widehat{\mathbf{\Theta}}_{\MAP}^{\mathcal{H}_0})
	= 1 - \Phi \left( \frac{\lambda - P_{t} + 10 \alpha \log \psi'_i } { \sigma_{\epsilon} } \right).
	\end{aligned}
	\end{equation}
Under $\mathcal{H}_1$, the adversary adds additional delays to the delay measurements such that the estimated distance to an anchor will be enlarged if the anchor receives a measurement and decreased if there is an inaudible scenario. The latter is due to the fact that the estimated target location will tend to be closer towards the inaudible anchors as illustrated in Fig.~\ref{fig:howitworks}.
As such, we obtain:
	\begin{equation}
	\begin{split}
	P(r_i = 0 | \widehat{\mathbf{\Theta}}_{\MAP}^{\mathcal{H}_1})
	= \Phi \left( \frac{\lambda - P_{t} + 10 \alpha \log (\psi'_i -\mu') } { \sigma_{\epsilon} } \right),
	\\
	P(r_i = 1 | \widehat{\mathbf{\Theta}}_{\MAP}^{\mathcal{H}_1})
	= 1 - \Phi \left( \frac{\lambda - P_{t} + 10 \alpha \log (\psi'_i +\mu') } { \sigma_{\epsilon} } \right).
	\end{split}
	\end{equation}
Substituting the above audibility terms into \eqref{eqn:teststataudibility} and taking logarithm on both sides, the test statistic becomes

	\begin{align*}
	&\Lambda({\bf t}, {\bf r})     = \sum^{l}_{i=1}  \ln  \frac{1}{\sqrt{2\pi} (\sigma_W + \sigma_{\delta})}
- \sum^{l}_{i=1} \frac{ (t_i - \psi_i - \mu_{\delta})^2 }{2 (\sigma^2_W + \sigma^2_{\delta}) }
	\\&
	+  \sum^{l}_{i=1} \ln                   \bigg(1 - \Phi \left( \frac{\lambda - P_{t} + 10 \alpha \log (\psi'_i +\mu') } { \sigma_{\epsilon} } \right) \bigg)
	\\&
	+ \sum^{n}_{i=l+1} \ln    \Phi \left( \frac{\lambda - P_{t} + 10 \alpha \log (\psi'_i -\mu') } { \sigma_{\epsilon} } \right)
	\\&
	-  \Bigg[ \sum^{l}_{i=1}   \ln  \frac{1}{\sqrt{2\pi} \sigma_W}
	- \sum^{l}_{i=1}     \frac{ (t_i - \psi_i)^2  }{2 \sigma^2_W }
	\\&
	+  \sum^{l}_{i=1} \ln                            \bigg(1 - \Phi \left( \frac{\lambda - P_{t} + 10 \alpha \log \psi'_i } { \sigma_{\epsilon} } \right) \bigg)
	\\&
	+ \sum^{n}_{i=l+1} \ln   \Phi \left( \frac{\lambda - P_{t} + 10 \alpha \log \psi'_i } { \sigma_{\epsilon} } \right)        \Bigg]
	\underset{\mathcal{H}_0}{\overset{\mathcal{H}_1}{\gtrless}}
 	 \ln \eta. \tageq
	\end{align*}

Next, we simplify and rearrange the terms to get
	\begin{align*}
	&\Lambda({\bf t}, {\bf r})
	\\&=  \ln \frac{ \sigma_W }{(\sigma_W + \sigma_{\delta})}
- \sum^{l}_{i=1} \frac{ (t_i - \psi_i - \mu_{\delta})^2 }{2 (\sigma^2_W + \sigma^2_{\delta}) }
	+ \sum^{l}_{i=1}     \frac{ (t_i - \psi_i)^2  }{2 \sigma^2_W }
	\\&
	+ \Bigg[     \sum^{l}_{i=1} \ln                    \bigg(1 - \Phi \left( \frac{\lambda - P_{t} + 10 \alpha \log (\psi'_i+\mu') } { \sigma_{\epsilon} } \right) \bigg)
	\\&
	+ \sum^{n}_{i=l+1} \ln    \Phi \left( \frac{\lambda - P_{t} + 10 \alpha \log (\psi'_i-\mu') } { \sigma_{\epsilon} } \right)
	\\&
	-  \sum^{l}_{i=1} \ln                            \bigg(1 - \Phi \left( \frac{\lambda - P_{t} + 10 \alpha \log \psi'_i } { \sigma_{\epsilon} } \right) \bigg)
	\\&
	- \sum^{n}_{i=l+1} \ln   \Phi \left( \frac{\lambda - P_{t} + 10 \alpha \log \psi'_i } { \sigma_{\epsilon} } \right)     	\Bigg]
	\underset{\mathcal{H}_0}{\overset{\mathcal{H}_1}{\gtrless}}
 	 \ln \eta. \tageq
\label{eqn:derive_xi}
	\end{align*}
Finally, we obtain
	\begin{align*}
	&\Lambda({\bf t}, {\bf r})    = \ln \frac{ \sigma_W }{(\sigma_W + \sigma_{\delta})}
- \sum^{l}_{i=1} \frac{ (t_i - \psi_i - \mu_{\delta})^2 }{2 (\sigma^2_W + \sigma^2_{\delta}) }
	+ \sum^{l}_{i=1}     \frac{ (t_i - \psi_i)^2  }{2 \sigma^2_W }
	\\&
\quad\quad\quad\quad + \sum^{n}_{i=1}\Xi_i
	\underset{\mathcal{H}_0}{\overset{\mathcal{H}_1}{\gtrless}}
 	2 \sigma^2_W \ln \eta,
\label{eqn:xi}
	\\&
 \sum^{l}_{i=1}
 \sigma^2_{\delta} t_i^2 + 2 \mu_{\delta} \sigma_W^2 t_i - 2\sigma^2_{\delta} \psi_i t_i
+ 2\sigma_W^2 (\sigma_W^2 +    \sigma^2_{\delta})  \sum^{n}_{i=1}\Xi_i
	\\& \quad
	\underset{\mathcal{H}_0}{\overset{\mathcal{H}_1}{\gtrless}}
 	2\sigma_W^2 (\sigma_W^2 +    \sigma^2_{\delta})
\ln  \left(
\frac{\eta(\sigma_W + \sigma_{\delta})  }{ \sigma_W}  \right)
	\\& \quad\quad
+ \sum^{l}_{i=1}    2 \psi_i \mu_{\delta} \sigma_W^2 + \mu_{\delta}^2 \sigma_W^2
- \sigma^2_{\delta} \psi_i^2,
	\\&
 \sum^{l}_{i=1}
 \sigma^2_{\delta} t_i^2 + 2 \mu_{\delta} \sigma_W^2 t_i - 2\sigma^2_{\delta} \psi_i t_i
+ 2\sigma_W^2 (\sigma_W^2 +    \sigma^2_{\delta})  \sum^{n}_{i=1}\Xi_i
	\underset{\mathcal{H}_0}{\overset{\mathcal{H}_1}{\gtrless}}
\gamma. \tageq
	\end{align*}
where $\Xi_i$ is some function of the audibility terms (fourth term of \eqref{eqn:derive_xi} in $[.]$ brackets) and is independent of the delay measurements ${\bf t}$.
Using the same $Z = \sum^{l}_{i=1}
 \sigma^2_{\delta} t_i^2 + 2 \mu_{\delta} \sigma_W^2 t_i - 2\sigma^2_{\delta} \psi_i t_i$
and  $\gamma$ as the previous Appendix-\ref{Appendix_woaudibility2}, and let \mbox{$\Xi = 2\sigma_W^2 (\sigma_W^2 +    \sigma^2_{\delta})  \sum^{n}_{i=1}\Xi_i$},
the detection probability for the audibility-aware GLRT is given by
	\begin{equation}
	\begin{split}
	P_{d | \psi}^{\rm A} &=
	P( z  + \Xi > \gamma  |  \mathcal{H}_1, \psi   )
	\\& =
	\int_\gamma^{\infty}    p(z  + \Xi  |  \mathcal{H}_1, \psi   ) dz,
	\end{split}
	\end{equation}
and the false alarm probability is given by
	\begin{equation}
	\begin{split}
	P_{f | \psi}^{\rm A} &=
	P( z > \gamma  |  \mathcal{H}_0, \psi   )
	\\& =
	\int_\gamma^{\infty}    p(z  |  \mathcal{H}_0, \psi   )   dz,
	\end{split}
	\end{equation}
as $\Xi = 0$ under $\mathcal{H}_0$ due to the audibility terms being canceled out by each other when $\mu_{\delta}=0$.

%\bibliographystyle{IEEEtran}%Choose a bibliograhpic style
%\bibliography{references}

\end{document}